\documentclass{amsart}
\usepackage{graphicx}    
\usepackage{subcaption}
\usepackage{mathrsfs}
\usepackage{amssymb}
\usepackage{amsbsy}
\usepackage{suffix}
\usepackage{mathtools}
\usepackage{booktabs}
\usepackage{diagbox}
\usepackage{upgreek}

\usepackage{color}

\newtheorem{proposition}{Proposition}[section]
\newtheorem{lemma}{Lemma}[section]

\theoremstyle{definition}

\newtheorem{example}{Example}[section]

\theoremstyle{remark}
\newtheorem{remark}{Remark}[section]

\numberwithin{equation}{section}

\def\XXint#1#2#3{{\setbox0=\hbox{$#1{#2#3}{\int}$}
\vcenter{\hbox{$#2#3$}}\kern-.5\wd0}}

\DeclarePairedDelimiterX\MeijerM[3]{\lparen}{\rparen}%
{\begin{smallmatrix}#1 \\ #2\end{smallmatrix}\delimsixe\vert\,#3}

\newcommand\MeijerG[8][]{%
  G^{\,#2,#3}_{#4,#5}\MeijerM[#1]{#6}{#7}{#8}}

\WithSuffix\newcommand\MeijerG*[7]{%
  G^{\,#1,#2}_{#3,#4}\MeijerM*{#5}{#6}{#7}}
\def\braket#1{\mathinner{\langle{#1}\rangle}}

\let\protect\relax
{\catcode`\|=\active
  \xdef\Braket{\protect\expandafter\noexpand\csname Braket \endcsname}
  \expandafter\gdef\csname Braket \endcsname#1{\begingroup
     \ifx\SavedDoubleVert\relax
       \let\SavedDoubleVert\|\let\|\BraDoubleVert
     \fi
     \mathcode`\|32768\let|\BraVert
     \left\langle{#1}\right\rangle\endgroup}
}
\def\BraVert{\@ifnextchar|{\|\@gobble}
     {\egroup\,\mid@vertical\,\bgroup}}
\def\BraDoubleVert{\egroup\,\mid@dblvertical\,\bgroup}
\let\SavedDoubleVert\relax

{\catcode`\|=\active
  \xdef\set{\protect\expandafter\noexpand\csname set \endcsname}
  \expandafter\gdef\csname set \endcsname#1{\mathinner
        {\lbrace\,{\mathcode`\|32768\let|\midvert #1}\,\rbrace}}
  \xdef\Set{\protect\expandafter\noexpand\csname Set \endcsname}
  \expandafter\gdef\csname Set \endcsname#1{\left\{%
     \ifx\SavedDoubleVert\relax \let\SavedDoubleVert\|\fi
     \:{\let\|\SetDoubleVert
     \mathcode`\|32768\let|\SetVert
     #1}\:\right\}}
}
\def\midvert{\egroup\mid\bgroup}
\def\SetVert{\@ifnextchar|{\|\@gobble}
    {\egroup\;\mid@vertical\;\bgroup}}
\def\SetDoubleVert{\egroup\;\mid@dblvertical\;\bgroup}

\begingroup
 \edef\@tempa{\meaning\middle}
 \edef\@tempb{\string\middle}
\expandafter \endgroup \ifx\@tempa\@tempb
 \def\mid@vertical{\middle|}
 \def\mid@dblvertical{\middle\SavedDoubleVert}
\else
 \def\mid@vertical{\mskip1mu\vrule\mskip1mu}
 \def\mid@dblvertical{\mskip1mu\vrule\mskip2.5mu\vrule\mskip1mu}
\fi

\begin{document}

\title[Products of matrices and renewing flows]{The continuum limit of some products of random matrices
associated with renewing flows}

\author{Yves Tourigny}
\address{School of Mathematics\\
        University of Bristol\\
        Bristol BS8 1UG, United Kingdom}
\email{y.tourigny@bristol.ac.uk}

\thanks{ORCID: 0000-0002-9646-0401. This work was initiated during a visit to the Laboratoire de Physique Th\'eorique et Mod\`eles Statistiques in the summer of 2023. I am grateful to the Universit\'{e} Paris--Saclay for supporting the visit financially, and to my long-time collaborators
Alain Comtet and Christophe Texier for many helpful discussions. The author has no relevant financial or non-financial interests to disclose. There is no data to be made available.}




\begin{abstract}
We consider the continuum limit of some products 
of random matrices in $\text{SL}(d,{\mathbb R})$ that arise as discretisations of incompressible renewing flows--- that is, of flows corresponding to a divergence-free velocity field that takes independent, identically-distributed
values in successive time intervals of duration proportional to $\tau$. The statistical properties of the product are encoded in its generalised Lyapunov exponent whose
computation reduces to finding the leading eigenvalue of a certain transfer operator.
In the continuum limit obtained by neglecting the terms of order $o(\tau^2)$, the transfer operator becomes a partial differential operator and, for a certain type of disorder which we call ``symmetric'', some calculations are feasible. For $d=2$,  we compute the growth rate of the product
in terms of complete elliptic integrals. By letting the elliptic modulus vary, we obtain a spectral problem, corresponding to a modulus-dependent random renewing flow, which may be viewed as a perturbation of the spectral problem for the angular Laplacian. In this way,
we deduce expansions for the generalised Lyapunov exponent in ascending powers of the modulus. These expansions generalise to the case $d \ge 2$, and we compute the first few terms explicitly for $d \in \{2,3\}$.
\end{abstract} 

\maketitle

\section{Introduction}

Products of random matrices arise as mathematical models of many disordered physical systems; see Part II of the monograph \cite{CPV}--- and the extensive bibliography therein--- where concrete applications to the study of Ising spin chains, transport in random media, chaotic dynamics etc. are described.
It is typical of such products that they grow at a deterministic, almost-sure rate as the number of elements in the product increases. This growth rate, known as the {\em Lyapunov exponent} of the product, is the most basic quantity of interest in the applications, as it usually provides a quantitative measure of some important aspect of the physical phenomenon being modelled; for instance the reciprocal of the localisation length in systems exhibiting Anderson localisation, or the rate at which nearby trajectories separate in a chaotic system.
However, the fact that the elements seldom commute makes its calculation challenging. A brief survey of the exact results that have been obtained in this direction may be found in \cite{CTT1}; such cases are few, and the calculation relies
on special features of the disorder--- as in Dyson's random chain model \cite{Dy} or Lloyd's tight-binding model \cite{Ll}--- and/or on taking the continuum limit of the product--- as in Halperin's study of a quantum particle in a white noise potential \cite{Ha}. The 
continuum limit of the two-by-two case is exceptional: the problem of computing the Lyapunov exponent was effectively solved in \cite{CLTT}.


If one looks beyond the almost-sure growth rate and seeks to ascertain the nature of the product's fluctuations then the relevant quantity is the {\em generalised Lyapunov exponent} \cite{CPV}, which will be defined shortly, and whose calculation presents an even greater challenge.
In \cite{CTT1}, we exhibited one case where detailed calculations were possible, namely the continuum limit of a product of random matrices in the group $\text{SL}(2,{\mathbb R})$ 
that models a quantum particle in a one-dimensional potential with impurities. It is the purpose of the present paper to describe another such case; only, this time,
the product under scrutiny concerns toy models for turbulent incompressible fluid flow in $d$-dimensional space, so that the matrices in the product belong to the group $\text{SL}(d,{\mathbb R})$. As we shall explain in the remainder of this introduction, the new case features the ingredients necessary to make progress:
(1) we consider a continuum limit of the product; (2) we assume that the disorder takes a highly symmetric form; (3) we embed the problem in a parameter-dependent family that contains a trivially solvable case.

\subsection{Motivation} 
\label{motivationSubsection}
Let us make clear at the outset that the results reported here make no direct contribution to fluid mechanics.
Nevertheless, since the project was initiated after reading the very interesting article of Haynes \& Vanneste \cite{HV}, it is proper to begin with a brief description of the physical context, as it provides both the motivation for considering these models, and the scope of the study to follow.

The diffusion of a scalar carried along by a flow in $d$-dimensional space is governed by the equation
\begin{equation}
\frac{\partial \rho}{\partial t} = u \cdot \nabla \rho + D\,\Delta \rho\,.
\label{advectionDiffusionEquation}
\end{equation}
Here, $\rho = \rho(x,t)$ denotes, say, a temperature or the density of some pollutant, $D$ is the diffusivity of the fluid in which the scalar is advected, and $u = u(x,t)$ is the fluid's velocity field. It is assumed that the scalar is ``passive'', meaning that its presence in the fluid does not affect the velocity field. 

To any given velocity field, one may associate the dynamical system
\begin{equation}
\dot{x} = {u} (x,t)\,, \;\; t > 0\,.
\label{dynamicalSystem}
\end{equation}
If we consider the paths followed by particles near a given trajectory $x=x(t)$ , we see that the separation, say $\delta x (t)$,  between the paths evolves locally according to the equations
\begin{equation}
\frac{d}{dt}  \delta x_i = \sum_{j=1}^d \frac{\partial u_i}{\partial x_j}  \delta x_j\,,\;\;1 \le i \le d \,. 
\label{tangentProcess}
\end{equation}
For the case where the velocity field is {\em random}, Haynes \& Vanneste examine the relationship between the large-deviation statistics of $\delta x (t)$ and the long-time behaviour of the two-point covariance function
\begin{equation}
C(x,x',t) := {\mathbb E} \left [ \rho(x,t) \rho (x',t) \right ]\,.
\label{covarianceFunction}
\end{equation}
More precisely, they consider divergence-free random velocity fields that are spatially homogeneous, so that the covariance function takes the form
\begin{equation}
C(x,x',t) = \Gamma (x'-x,t)
\label{gammaFunction}
\end{equation}
for some function $\Gamma$ of space and time; their aim is then to ascertain under what conditions knowledge of the flow's stretching characteristics suffices to determine the decay rate of $\Gamma(0,t)$ as $t \rightarrow \infty$. To this end, they look for velocity fields such that some closed equation for the temporal evolution of $\Gamma (x,t)$ may be derived and analysed. To construct velocity fields that meet these requirements, it is  convenient to assume that the scalar evolves in a spatial domain obtained by extending periodically a basic hypercube aligned with the axes, denoting by $P$ the (common) period in each of the coordinates.

\subsection{Kraichnan's velocity field ensemble}
\label{kraichnanSubsection}
With this motivation, Chetrite {\em et al.} \cite{CDG} consider a zero-mean Gaussian random velocity field process whose correlation matrix has entries of the form
\begin{multline}
{\mathbb E} \left [ u_i (x,t) \,u_j (x',t') \right ] \\
= \delta (t-t') \sum_{{\mathbf k} \in \frac{2 \pi}{P} {\mathbb Z}^d} \left [ (1-\vartheta) \,\delta_{ij} - (1-\vartheta d) \,k_{i} k_j \right ] e^{\text{\rm i} {\mathbf k} \cdot (x-x')} \varrho (|{\mathbf k}|)
\label{kraichnanCorrelation}
\end{multline}
for $1 \le i,j \le d$. Here $\varrho$ is some spectral density function, and  $0 \le \vartheta \le 1$ is a parameter, called ``degree of compressibility'', which equals zero if and only if the velocity field is divergence-free. The model is designed so that the spatial statistics of the resulting velocity field remain the same under every transformation of space that leaves the hypercube invariant; this has the remarkable consequence that, regardless of the dimension $d$, these statistics are characterised by just three parameters, denoted $\alpha$, $\beta$ and $\gamma$ by Chetrite {\em et al.}, and constrained by the inequalities
$$
\gamma \ge | \beta | \;\;\text{and}\;\;4 \alpha + (d+2) \beta + 2 \gamma \ge d | \beta |\,.
$$
The degree of compressibility may be expressed as
\begin{equation}
\vartheta = \frac{2 \alpha + (d+1) \beta + 2 \gamma}{2 \alpha + 2 \beta + d \gamma}
\label{incompressibilityDegree}
\end{equation}
so that, in the case where the velocity field is divergence free, only two ``free'' parameters remain.
Chetrite {\em et al.} also define the ``degree of anisotropy'' of the random flow by
\begin{equation}
[0,1] \ni \kappa := \left | \frac{\alpha}{\alpha + \beta + \gamma} \right |\,.
\label{anisotropyDegree}
\end{equation}
This terminology is justified by the fact that the spatial statistics of the velocity field are invariant under {\em every} rotation if and only if $\alpha=0$.

With this model, Equation (\ref{tangentProcess}) becomes a diffusion process on the group--- $\text{SL} (d,{\mathbb R})$ if $\vartheta=0$, $\text{GL} (d,{\mathbb R})$ otherwise--- that can be studied via its generator. This generator takes the form of a second-order partial differential operator in the variables that are used to parametrise the group. We shall return to Chetrite {\em et al.}'s  analysis of this model in due course.

\subsection{Renewing flows}
\label{renewingFlowsSubsection}
In the present paper, we study an alternative model. We begin by considering a $d$-dimensional random velocity field with components of the form
\begin{equation}
u_i (x,t) := \sum_{j \ne i} u_{ij} \left ( x_j + \eta_{ij} (t) \right )\,,\;\;1 \le i \le d\,,
\label{velocityField}
\end{equation}
where
\begin{enumerate}
\item the $u_{ij} : {\mathbb R} \rightarrow {\mathbb R}$ are periodic functions with a common period $P/m$ for some $m \in {\mathbb N}$, and zero mean over a period;
\item the $\{\eta_{ij} (t) \}_{t \ge 0}$ are independent copies of the same random process
$\{ \eta (t)\}_{t>0}$;
\item there exists some time $\delta t >0$ such that, for every $t,\,t' \ge 0$, the inequality $|t - t'| \ge \delta t$
implies that the random variables $\eta (t)$ and $\eta(t')$ are independent.
\end{enumerate}
After selecting a suitable reference length, the common period of the $u_{ij}$ may be taken to equal $2 \pi$ in dimensionless units. We shall henceforth assume that, for every $t > 0$, the random variable $\eta(t)$ is uniformly distributed in $[ 0,2 \pi)$; this ensures that the resulting velocity field process is spatially homogeneous.

Let us then discretise the resulting continuous dynamical system in the following way:
\begin{align*}
x_d^{n+1} &=  x_d^n + \delta t \sum_{j=1}^{d-1} u_{dj} \left ( x_j^n + \eta_{dj}^n \right ) \\
x_{d-1}^{n+1} &=  x_{d-1}^n + \delta t \sum_{j=1}^{d-2} u_{d-1,j} \left ( x_j^n + \eta_{d-1,j}^{n} \right ) + \delta t \,u_{d-1,d} \left ( x_{d}^{n+1} + \eta_{d-1,d}^{n} \right ) \\
 & \vdots \\
x_1^{n+1}  &= x_1^n + \delta t \sum_{j=2}^d u_{1j}  \left ( x_j^{n+1} + \eta_{1j}^{n} \right )
\end{align*}
where $n$ runs over the natural numbers, so that $x^n$ is the value of $x$ at time $t = t_n := n \delta t \,d$. This is equivalent to considering a piecewise-constant-in-time velocity field such that, in each of the subintervals 
\begin{equation}
{\mathcal I}_{n,i} := \left ( t_{n+1}-i \,\delta t,\,t_{n+1}-(i-1) \,\delta t \right )\,, \;\;i \in \{ 1,\,2,\,\ldots\,, d\}\,,
\label{timeInterval}
\end{equation}
the only non-vanishing velocity component  is the $i$th one.  Such flows are called ``random flows with renewal'' or, simply, ``renewing flows'' \cite{CG,ZMRS}. 

As a concrete example, for $d=2$, one recovers the ``sine flow'' example considered by Haynes \& Vanneste--- see \cite{HV}, Equations (3.17) and (4.8)--- 
by taking both $u_{12}$ and $u_{21}$ to be some multiple of the sine function; the resulting equation for $\Gamma (x,t)$ is derived in their \S 4. Other examples may be found in \cite{CFV}.

\subsection{Products of matrices from renewing flows}
\label{flowToProductSubsection}
The map $x^n \mapsto x^{n+1}$ defines a discrete dynamical system and, in this setting,
the separation between nearby trajectories is governed by the product of the successive Jacobian matrices, indexed by $n$ and evaluated at $x_n$.
For instance, in the case $d=2$, the Jacobian matrix, evaluated at $x_n$, is
\begin{equation}
\begin{pmatrix}
1 & \delta t\, u_{12}' \left ( x_2^{n+1} + \eta_{12}^n \right ) \\
0 & 1
\end{pmatrix}
\begin{pmatrix}
1 & 0 \\
\delta t\, u_{21}' \left ( x_1^n + \eta_{21}^n \right ) & 1
\end{pmatrix}
\label{2by2Matrix}
\end{equation}
where the $\eta_{ij}^n$ are independent draws from the uniform distribution on $(-\pi,\pi)$. 
An important consequence of the hypotheses made is that
the distribution of the Jacobian matrix, conditional on the event $x^n = x$, is independent of $x$. We may therefore replace the Jacobian matrix by the value, say $g_n$, that it takes at $x^n = 0$:
\begin{equation}
g_n := \begin{pmatrix} 1 & \delta t \,u_{12}' \left ( \phi_{12}^n  \right ) \\
0 & 1
\end{pmatrix}
\begin{pmatrix}
1 & 0 \\
\delta t \,u_{21}' \left (\phi_{21}^n \right ) & 1
\end{pmatrix}
\label{2dRandomMatrix}
\end{equation}
where
$$
\phi_{12} := \eta_{12} + \delta t\, u_{21} \left ( \eta_{21} \right )\;\;\text{and}\;\;\phi_{21} = \eta_{21}
$$
so that $g_1$, $g_2$, \dots are independent and identically distributed. 

The same reasoning applies in the higher-dimensional case; the Jacobian matrix of the map 
$x^n \mapsto x^{n+1}$, evaluated at $x^n=0$, has the same distribution as the matrix
\begin{equation}
g_n := \prod_{j \ne 1} e^{\delta t \,\xi_{1j}^n N_{1j}} \cdots \prod_{j \ne d} e^{\delta t \,\xi_{dj}^n N_{dj}}
\label{jacobian}
\end{equation}
where $N_{ij}$ denotes the $d \times d$ matrix whose only non-zero entry is a one in the $i$th row and $j$th column,
\begin{equation}
\xi_{ij} := u_{ij}' \left ( \phi_{ij} \right )\,,\;\;1 \le i \ne j \le d\,,
\label{tij}
\end{equation}
\begin{equation}
\phi_{ij} := \eta_{ij}\,, \;\;1 \le j < i \le d\,,
\label{lowerPhiAngles}
\end{equation}
and
\begin{equation}
\phi_{ij} := \eta_{ij} + \delta t\, \sum_{k \ne j} u_{jk} \left (  \phi_{jk} \right ) , \;\;1 \le i < j \le d\,.
\label{upperPhiAngles}
\end{equation}
To clarify the meaning of Equation (\ref{jacobian}): the order in which the $d$ products appear is important, but the order in which the terms in each product are multiplied is not because, for $i$ fixed, the $e^{\delta t \xi_{ij} N_{ij}}$, $j \ne i$, commute
amongst themselves.

It is proved  in Appendix \ref{propositionAppendix} that  the $\xi_{ij}$ are {\em independent}.
We are thus led to studying certain products
\begin{equation}
\Pi_n := g_n \cdots g_2 g_1
\label{productOfMatrices}
\end{equation}
of random independent and identically-distributed matrices with a particular structure: Firstly, the $g_n$ have unit determinant; this is a consequence of the fact that the velocity field defined by Equation (\ref{velocityField}) is divergence-free. Secondly, because of the way we chose to discretise the dynamical system,
each $g_n$ is itself a product of random independent triangular
matrices; this will facilitate the analysis to follow. 


\subsection{The generalised Lyapunov exponent}
\label{GLEsubsection}
Associated with the product (\ref{productOfMatrices}) is the {\em generalised Lyapunov exponent}:
\begin{equation}
L (\ell) := \lim_{n \rightarrow \infty} \frac{1}{n} \ln {\mathbb E} \left [ 1_\ell \left (x \Pi_n \right ) \right ] = \sum_{j=1}^\infty  c_j \frac{\ell^j}{j!}
\label{generalisedLyapunovExponent}
\end{equation}
where $x \in {\mathbb R}_\ast^d := {\mathbb R}^d \backslash \{0\}$ is interpreted as a row vector, and $1_\ell$ is the function defined on
${\mathbb R}_\ast^d$ by 
\begin{equation}
1_{\ell} (x) := \left | x \right |^{\ell} = \left ( x_1^2 + \cdots + x_d^2 \right )^\frac{\ell}{2}\,.
\label{unitFunction}
\end{equation}

$L$ may be viewed as a {\em scaled cumulant generating function} associated with the dimensionless random variable $\ln \left | {x} \Pi_n \right |$; it is in fact independent of the ${x}$ appearing in its definition. In terms of the renewing flow, $c_j/(\delta t \,d)$ is the discrete counterpart of
$$
\lim_{t \rightarrow \infty} \frac{1}{t} \ln {\mathbb E} \left \{ \left [ \ln \frac{\left | \delta x (t) \right |}{\left | \delta x (0) \right |}\right ]^j\right \}
$$
so that, in particular, $c_1/(\delta t\,d)$ measures the rate at which nearby trajectories separate. More generally, $L$ encodes the large-$n$ statistical properties of the product: if $\text{\rm d} s$ 
denotes an interval of infinitesimal length $ds$ centered at $s \in {\mathbb R}$, then
\begin{equation}
{\mathbb P} \left ( \frac{1}{n} \ln \left | x \Pi_n \right | \in \text{\rm d} s \right )  \propto \exp \left [ - n \,I(s) \right ]\,d s\;\;\text{as $n \rightarrow \infty$}\,.
\label{largeDeviation}
\end{equation}
where $I$, known variously as the ``large deviation'' \cite{De} or ``rate'' \cite{TH} or ``entropy'' \cite{CPV} or ``Cram\'{e}r'' \cite{HV} function associated with the product,
is the Legendre transform of $L$:
\begin{equation}
I (s) := \min_{\ell \in {\mathbb R}} \left [ \ell s - L (\ell) \right ]\,.
\label{rateFunction}
\end{equation}

\subsection{The transfer operator}
\label{transferOperatorSubsection}
To compute the generalised Lyapunov exponent, 
we first observe that the function $1_\ell$ defined by Equation (\ref{unitFunction}) is even and homogeneous of degree $\ell$--- namely,
for every non-zero number $r$,
$$
1_\ell (r x) = | r |^{\ell} 1_\ell ( x )\,.
$$
We denote by $V_\ell$ the space of even smooth functions on ${\mathbb R}_\ast^d$, homogeneous of degree $\ell$, and define a representation $T_\ell$ of the group $\text{SL}(d,{\mathbb R})$ with representation space $V_\ell$ by the formula
\begin{equation}
\left [ T_\ell (g) v \right ] (x) = v ( x g )\,.
\label{basicRepresentation}
\end{equation}
By using the fact that the $g_n$ are independent and identically distributed, it may be verified that
$$
{\mathbb E} \left [ \left | x \Pi_n \right |^{\ell} \right ] =  [ \underbrace{{\mathscr T}_\ell \circ \cdots \circ {\mathscr T}_\ell}_{\text{$n$ times}} 1_\ell  ] (x)
$$
where 
\begin{equation}
{\mathscr T}_\ell := {\mathbb E} \left [ T_\ell (g) \right ]
\label{transferOperator}
\end{equation}
is the {\em transfer operator}.

The rigorous mathematical treatment of the spectral properties of (non-self-adjoint, non-compact) linear operators in infinite-dimensional spaces is a difficult undertaking that involves subtle technical concepts. However, since our objective is of a practical, computational nature, it will be sufficient for our purpose to state the eigenvalue problem for the transfer operator naively as: find numbers $\lambda$ such that there exists a non-zero solution $v \in V_\ell$ of the equation
\begin{equation}
{\mathscr T}_\ell v = \lambda v\,.
\label{eigenvalueProblem}
\end{equation}
It then follows from the foregoing that 
\begin{equation}
L ( \ell ) = \ln \lambda ( \ell)
\label{spectralCharacterisation}
\end{equation}
where $\lambda (\ell)$ is the dominant eigenvalue of the transfer operator. This characterisation of the generalised Lyapunov exponent assumes that there is a ``spectral gap''--- that is, that the dominant eigenvalue is well separated from the rest of the spectrum.

\subsection{The continuum limit}
For the piecewise-constant-in-time velocity field considered in this paper, only the diagonal entries of the velocity field's correlation matrix are non-zero; a simple calculation shows that, for $1 \le i \le d$, they take the form
\begin{equation}
{\mathbb E} \left [ u_i (x,t) \,u_i (x',t') \right ] \\
=   \sum_{j \ne i} b_{ij} (x_j-x_j') \,\begin{cases} 1 & \text{if $t,\,t' \in {\mathcal I}_{n,i}$ for some $n \in {\mathbb N}$} \\ 0 & \text{otherwise} \end{cases}
\label{correlationFunction}
\end{equation}
where ${\mathcal I}_{n,i}$ is the time interval defined by Equation (\ref{timeInterval}) and the $b_{ij}$ are cosine series with positive coefficients.
If we multiply this identity by a suitable test function of $t'$ and integrate over $t'$, we see that the limit $\delta t \rightarrow 0+$, henceforth referred to as the {\em continuum limit}, corresponds to a velocity field that, as in the Kraichnan ensemble of \S \ref{kraichnanSubsection}, is delta-correlated in time.

Now, by virtue of Equation (\ref{jacobian}) and the independence of the $\xi_{ij}$, the transfer operator (\ref{transferOperator})
factorises as
\begin{equation}
{\mathscr T}_\ell =
\prod_{j \ne 1} {\mathbb E} \left ( e^{\delta t \,\xi_{1j}^n N_{1j}} \right ) \circ \cdots \circ \prod_{j \ne d} {\mathbb E} \left ( e^{\delta t \,\xi_{dj}^n N_{dj}} \right )\,.
\label{factorisation}
\end{equation}
In this formula, the elements $N_{ij}$ of the Lie algebra of $\text{\rm SL} (d,{\mathbb R})$ are now realised
as first-order differential operators associated with the representation $T_\ell$.
Expanding the transfer operator
(\ref{transferOperator}) in powers of $\delta t$, we arrive at
\begin{equation}
{\mathscr T}_\ell = 1 + \frac{\delta t^2}{2} \sum_{i \ne j} {\mathbb E} \left ( \xi_{ij}^2 \right ) N_{ij}^2
+ O \left ( \delta t^4 \right )\;\;\text{as $\delta t \rightarrow 0+$}\,.
\label{continuumLimit}
\end{equation}
Here, we have used the facts that, firstly, the $\xi_{ij}$ are mutually independent and, secondly,
$$
{\mathbb E} \left ( \xi_{ij} \right ) = \frac{1}{2 \pi} \int_{-\pi}^\pi d \theta \,u_{ij}' (\theta)
= \frac{1}{2 \pi} u_{ij} (\theta) \Bigl |_{-\pi}^\pi = 0
$$
since, by hypothesis, the $u_{ij}$ are $2 \pi$-periodic.

We shall say that the model exhibits {\em symmetric disorder} if there exists $\sigma>0$ such that
\begin{equation}
\forall\,\, 1 \le i \ne j \le d\,,\;\;{\mathbb E} \left ( \xi_{ij}^2 \right ) := \frac{1}{2 \pi} \int_{-\pi}^\pi d \theta\,\left [ u_{ij}' (\theta) \right ]^2 = \sigma^2\,.
\label{variance}
\end{equation}
Taking $\sqrt{2}/\sigma$ as our reference time, we see that, for symmetric disorder,
the continuum limit of the transfer operator is 
\begin{equation}
{\mathscr T}_\ell = 1 + \tau^2 \sum_{i \ne j} N_{ij}^2
\label{symmetricContinuumLimit}
\end{equation}
where
\begin{equation}
\tau := \frac{\sigma}{\sqrt{2}}\, \delta t
\label{dimensionlessParameter}
\end{equation}
is a dimensionless parameter, assumed small.


\subsection{Outline of the paper and main results}
\label{outlineSubsection}

In \S \ref{infinitesimalSection} we introduce the mathematical concepts we need to state and prove the
\begin{proposition}
\begin{equation}
\notag
\sum_{i \ne j} N_{ij}^2 = \Delta_K +  (d-1)\,\ell + \frac{d-1}{d} \,\ell^2 - \frac{1}{d} \sum_{i<j} A_{ij}^2
\end{equation}
where $\Delta_K$ is the Casimir of the subgroup $\text{\rm SO}(d,{\mathbb R})$ defined by
$$
\Delta_K := \sum_{i<j} K_{ij}^2
$$
and the $A_{ij}$, given by Formula (\ref{AinfinitesimalGenerators}), are the infinitesimal generators associated with the one-parameter subgroups of $\text{\rm SL} (d,{\mathbb R})$ consisting of diagonal matrices.
\label{casimirProposition}
\end{proposition}

Let us elaborate the interest of this elementary proposition. Fix $k \ge 0$ and consider the companion spectral problem: Find $\mu$ such that there exists a non-zero solution $v \in V_\ell$
of
\begin{equation}
 \mu v =  \left ( \Delta_K - k^2 \sum_{i<j} A_{ij}^2 \right ) v\,.
 \label{perturbedSpectralProblem}
\end{equation}
Equation (\ref{spectralCharacterisation}) then implies that
\begin{equation}
L (\ell) = \sum_{j=1}^\infty c_j  \frac{\ell^j}{j!} = \tau^2 \Lambda \left ( k ,\ell \right ) {\Bigl |_{k=1/\sqrt{d}}} + o \left ( \tau^2 \right )\;\;\text{as $\tau \rightarrow 0+$}
\label{generalisedLyapunovExpansion}
\end{equation}
where
\begin{equation}
\Lambda(k,\ell) = \sum_{j=1}^\infty \gamma_j (k) \frac{\ell^j}{j!}:= (d-1) \,\ell + \frac{d-1}{d} \,\ell^2 + \mu(k,\ell) 
\label{kDependentGeneralisedLyapunovExpansion}
\end{equation}
and $\mu(k,\ell)$ is the particular eigenvalue that maximises the expression on the right-hand side. We denote by
\begin{equation}
\Upsilon (k,s) := \min_{\ell \in {\mathbb R}} \left [ \ell s - \Lambda (k,\ell) \right ]
\label{legendreTransform}
\end{equation}
the Legendre transform of $\Lambda(k,\ell)$.
We then take advantage of the fact that the companion spectral problem (\ref{perturbedSpectralProblem}) is completely solvable when $k=0$ by embedding 
our original problem in a parameter-dependent family, where the parameter $k$ assumes ``small'' values. To facilitate the calculations, we work with the Iwasawa realisation, described in \S \ref{iwasawaSection}, of the basic representation, because it is in this realisation that the eigenfunctions of the Laplacian assume their simplest form.

In \S \ref{continuumLimit2dSection}, we study the $k$-dependent spectral problem (\ref{perturbedSpectralProblem}) when $d=2$. First, for fixed $0 \le k < 1$, we consider its small-$\ell$ limit. By making a change of variable, we obtain a problem 
that is solvable when $\ell=0$ and this eventually leads to explicit formulae for the first two cumulants:
\begin{equation}
\gamma_1 (k)=  2 \,\frac{\boldsymbol{\mathsf{E}}}{\boldsymbol{\mathsf{K}}} - 1
\label{growthRateForDis2}
\end{equation}
and
\begin{equation}
\frac{\gamma_2(k)}{2} =  \frac{3}{2}- k^2 - \frac{\boldsymbol{\mathsf{E}}}{\boldsymbol{\mathsf{K}}} + \frac{2 \pi^2}{\boldsymbol{\mathsf{K}}^2} \sum_{j=1}^\infty \left ( \frac{q^{2j}}{1-q^{4j}} \right )^2
- \frac{6 \pi^2}{\boldsymbol{\mathsf{K}}^2} \sum_{j=1}^\infty \left ( \frac{q^{2j-1}}{1-q^{4j-2}} \right )^2
\label{varianceForDis2}
\end{equation}
where $\boldsymbol{\mathsf{K}}$ and $\boldsymbol{\mathsf{E}}$ are the complete elliptic integrals of the first and second kind respectively, evaluated at $k$, and $q$ is the corresponding nome; see Equation (\ref{nome}). 

Equations (\ref{growthRateForDis2}-\ref{varianceForDis2}) essentially reproduce formulae obtained previously by Chetrite {\em et al.} for the two-dimensional Kraichnan model; see their sections 4.4 and 4.5 \cite{CDG}. Interestingly, Formula (\ref{growthRateForDis2}) also appears in connection with Anderson's model at zero energy and the Dirac equation with a random mass \cite{DG,ST,RT}. On the purely mathematical level, the study of the large-deviation statistics in all these models reduces to the same spectral problem
(\ref{perturbedSpectralProblem}); as we shall see, the only difference is that, in our model, the parameter $k$ assumes {\em real} rather than {\em imaginary} values. 

The fact 
that the parameter $k$ appears in formulae (\ref{growthRateForDis2}-\ref{varianceForDis2}) as the {modulus} of the elliptic integrals gives some confidence that its introduction may indeed be the ``right'' way of relating our spectral problem to that for $\Delta_K$.
Plots of 
\begin{equation}
\ell \mapsto \ell \,\gamma_1 \left ( k \right ) + \ell^2 \,\frac{\gamma_2 \left ( k \right )}{2} = \Lambda \left ( k,\ell \right ) + o \left ( \ell^2 \right )
\label{quadraticL}
\end{equation}
and of its Legendre transform
\begin{equation}
s \mapsto \frac{\left [ s-\gamma_1 \left ( k \right )\right ]^2}{2 \gamma_2 \left ( k \right )} = \Upsilon \left ( k, s \right ) + o \left ( s^2 \right )
\label{quadraticI}
\end{equation}
against their respective argument, for the case $k=1/\sqrt{2}$, are shown in Figure \ref{LFigure} (a). These formulae express the fact that, in the continuum limit, the distribution of the random variable
$\left | x \Pi_n \right |$ is approximately lognormal for large $n$ \cite{CPV}.

\begin{figure}[htbp]
  \centering
  \includegraphics[width=0.45\linewidth]{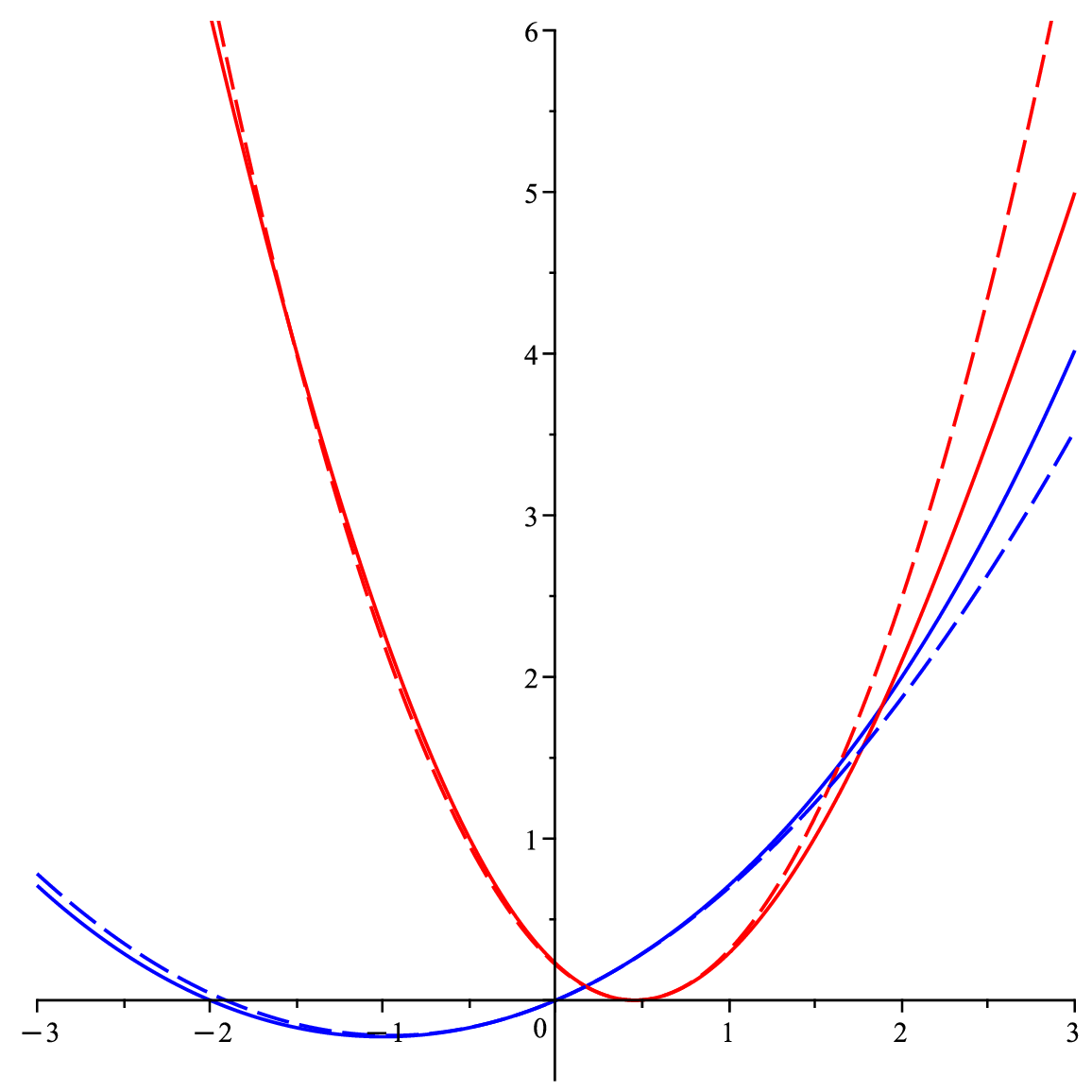}%
  \hspace{0.05\linewidth}%
  \includegraphics[width=0.45\linewidth]{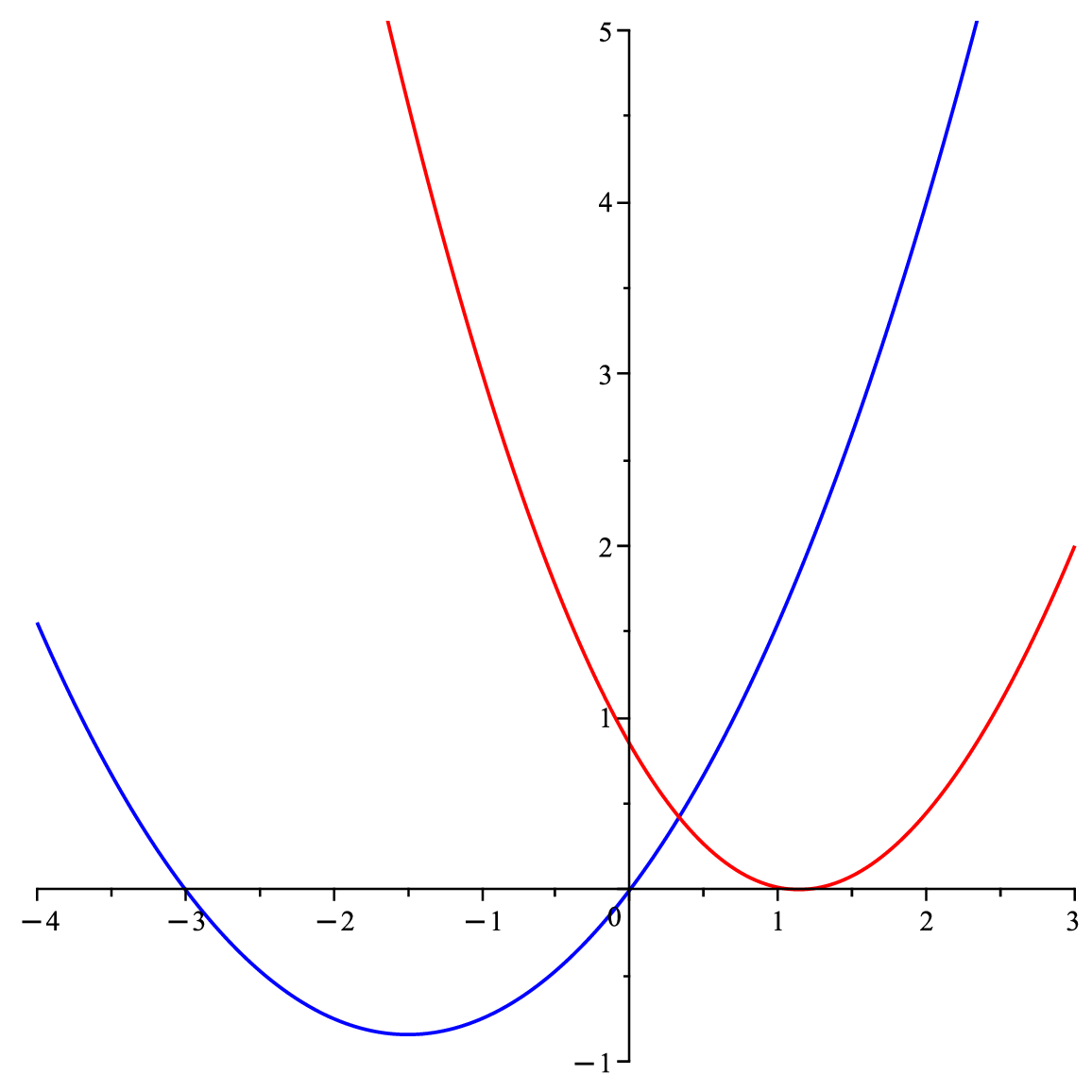}
  \begin{picture}(0,0) 
\put(-270,-7){(a)}
\put(-80,-7){(b)} 
\end{picture} 
  \caption{Plots of $\Lambda(1/\sqrt{d},\ell)$ against $\ell$ (solid blue curve) and $\Upsilon(1/\sqrt{d},s)$ against $s$
  (solid red curve) for (a) $d=2$ and (b) $d=3$. In (a) are also displayed the approximation (\ref{quadraticL})
  of $\Lambda(1/\sqrt{2},\ell)$ (dashed curve) and the approximation (\ref{quadraticI}) of $\Upsilon(1/\sqrt{2},s)$ (dot-dashed curve)
    which use only the first two cumulants.}
  \label{LFigure}
\end{figure}

Then, by considering the small-$k^2$ limit of the spectral problem (\ref{perturbedSpectralProblem}), we obtain expansions for $\Lambda$ and $\Upsilon$ in ascending powers of $k^2$: 
\begin{multline}
\Lambda(k,\ell) = \frac{\ell \left(\ell +2\right)}{2} -\frac{\ell \left(\ell +2\right)}{2} k^2  \\ + \frac{\left(\ell -2\right) \ell  \left(\ell +2\right) \left(\ell +4\right)}{128} k^4 
+ \frac{\left(\ell -2\right) \ell  \left(\ell +2\right) \left(\ell +4\right)}{256} k^6 \\
-\frac{\left(\ell -2\right) \ell  \left(\ell +2\right) \left(\ell +4\right)  \left(7 \ell^{4}+28 \ell^{3}+36 \ell^{2}+16 \ell -5248\right)}{2097152} k^8 \\
-\frac{\left(\ell -2\right) \ell \left(\ell +2\right) \left(\ell +4\right)  \left(21 \ell^{4}+84 \ell^{3}+108 \ell^{2}+48 \ell -7552\right)}{4194304} k^{10} + \cdots
\label{expansionForCGF2}
\end{multline}
and
\begin{multline}
\Upsilon(k,s) = \frac{\left(s -1\right)^{2}}{2} +\frac{\left(s -1\right) \left(s +1\right)}{2} k^2  \\  - \frac{s^4 -74 s^2 +9}{128} k^4  -\frac{\left ( 3 s^{2}-14 s +3\right ) \left ( 3 s^{2}+14 s +3\right )}{256} k^6 \\
+\frac{7 s^{8}+948 s^{6}-211974 s^{4}+1781972 s^{2}-47241}{2097152} k^8 \\
+\frac{133 s^{8}+15244 s^{6}-990274 s^{4}+4389324 s^{2}-67995}{4194304} k^{10} + \cdots
\label{expansionForRateFunction2}
\end{multline}

In \S \ref{continuumLimit3dSection}, we turn to the case $d=3$. In this case, there is no obvious change of variables that makes the problem tractable when $\ell=0$, and so we are unable to make progress in the small-$\ell$ limit. In the small-$k^2$ limit, however,  we can once again compute expansions 
for $\Lambda$ and its Legendre transform $\Upsilon$; we obtain
\begin{multline}
\Lambda(k,\ell) = \frac{2 \ell \left(\ell +3\right)}{3} -\frac{4 \ell \left(\ell +3\right)}{5} k^2  \\ + \frac{12 \left( \ell -2\right) \ell  \left(\ell +3\right) \left( \ell +5\right)}{875} k^4 \\
-\frac{72  \left(\ell -2\right) \ell \left(\ell +3\right) \left(\ell +5\right)  \left(9 \ell^{2}+27 \ell -395\right)}{3128125} k^6 + \cdots
\label{expansionForCGF3}
\end{multline}
and
\begin{multline}
\Upsilon(k,s) = \frac{3 \left(s -2\right)^{2}}{8} +\frac{9 \left(s -2\right) \left(s +2\right)}{20} k^2   \\
-\frac{243 s^{4}-36504 s^{2}+21168}{56000} k^4  \\ +\frac{59049 s^{6}-39724668 s^{4}+1597795632 s^{2}-421920576}{1601600000} k^6 + \cdots
\label{expansionForRateFunction3}
\end{multline}
Plots of $\Lambda(1/\sqrt{3}, \ell)$ against $\ell$ and $\Upsilon(1/\sqrt{3},s)$ against $s$ are shown in Figure \ref{LFigure} (b).

In \S \ref{pureStrainSection}, we show that, for $0 \le k^2 \le  1/d$, the function $\ell \mapsto \tau^2 \Lambda(k,\ell)$ is of intrinsic interest as the generalised Lyapunov exponent of (the continuum limit of) a product of random matrices associated with a renewing flow. Roughly speaking,
this renewing flow is obtained from the original one by inserting time intervals of pure strain in the $x_i-x_j$ planes, $i<j$; the number $1/d-k^2$ may then be interpreted as a measure of the strength of the disorder associated with the pure-straining, relative to the disorder in the original velocity gradient.

We conclude in \S \ref{concludingSection} with some brief comments on the possibility of developing the material in various directions.

\section{Infinitesimal generators}
\label{infinitesimalSection}

\subsection{The Lie algebra ${\mathfrak sl}(d)$}
\label{subgroupSubsection}
For the explicit construction and the analysis of the transfer operator, the following subgroups
will be particularly important:
\begin{enumerate}
\item {\em The compact subgroup $K := \text{\rm SO}(d,{\mathbb R})$} of unimodular orthogonal matrices. For $1 \le i < j \le d$, we denote by $k_{ij} (t)$
the element of $K$ corresponding to a rotation of angle $t$ in the $x_i-x_j$ plane. We can write
$$
k_{ij} (t) = e^{-t K_{ij}}
$$
for some $K_{ij}$ in the Lie algebra ${\mathfrak K}$ of $K$. It is clear that the $K_{ij}$ form a basis of 
${\mathfrak K}$. We have
$$
\dim {\mathfrak K} = \frac{d(d-1)}{2}\,.
$$
\item {\em The Abelian subgroup $A$} consisting of diagonal matrices of unit determinant
with positive entries along the diagonal. For $1 \le i < j \le d$, we denote by 
$$
a_{ij} (t) = e^{t A_{ij}}
$$ 
the one-parameter subgroup consisting of the diagonal matrices
with ones along the diagonal, except for the $i$th and $j$th entries, which are $e^{t}$ and $e^{-t}$
respectively. Clearly, the $A_{id}$, $1 \le i < d$, form a basis for the Lie algebra ${\mathfrak A}$ of $A$ and
$$
\dim {\mathfrak A} = d-1\,.
$$
\item {\em The nilpotent subgroups $N_-$ and $N_+$}  both consist of triangular
matrices with ones along the diagonal; the subscript $+$ (respectively $-$) being used for the subgroup of upper (respectively lower) triangular matrices. For $i \ne j$ we denote by 
$$
n_{ij} (t) = e^{t N_{ij}}
$$
the one-parameter subgroup consisting of those matrices in $\text{\rm SL} (d,{\mathbb R})$
with ones along the diagonal whose only non-zero off-diagonal entry is a $t$ in the $i$th row and $j$th column.
The $N_{ij}$ with $i<j$ (respectively $i>j$) form a basis of the Lie algebra ${\mathfrak N}^{+}$ (respectively ${\mathfrak N}^{-}$) of $N^{+}$ (respectively $N^-$), and we have
$$
\dim {\mathfrak N}^{\pm} = \frac{d (d-1)}{2}\,.
$$
\end{enumerate}

We can associate with each one-parameter subgroup $t \mapsto e^{tX}$
a differential operator acting on functions $v \in V_{\ell}$ via
$$
v \mapsto \frac{d}{d t} \left [ T_{\ell} \left ( e^{t X} \right ) v \right ] \Bigl |_{t=0}\,.
$$
This operator is called the {\em infinitesimal generator} of the representation $T_{\ell}$ associated with the one-parameter subgroup; it will be convenient to denote it also by the letter $X$.
The set of all such infinitesimal generators provides a realisation of the Lie algebra ${\mathfrak sl}(d)$ of $\text{SL}(d,{\mathbb R})$ in terms of first-order differential operators
in the $d$ coordinates $x_{j}$, $1 \le j \le d$, of ${x} \in {\mathbb R}_\ast^d$:
\begin{equation}
K_{ij} =   x_{i} \frac{\partial}{\partial x_{j}} - x_{j} \frac{\partial}{\partial x_{i}}\,,\;\;1 \le i < j \le d\,,
\label{KinfinitesimalGenerators}
\end{equation}
\begin{equation}
A_{ij} = x_{i} \frac{\partial}{\partial x_{i}} - x_{j} \frac{\partial}{\partial x_{j}}\,,\;\;1 \le i < j \le d\,,
\label{AinfinitesimalGenerators}
\end{equation}
and
\begin{equation}
N_{ij} = x_{i} \frac{\partial}{\partial x_{j}} \,,\;\;1 \le j \ne i \le d\,.
\label{NinfinitesimalGenerators}
\end{equation}

\subsection{The operator $\sum_{i \ne j} N_{ij}^2$ as a perturbation of the Casimir of $SO(d,{\mathbb R})$}
\label{casimirPropositionSection}
We are now ready to carry out the proof of Proposition \ref{casimirProposition}.
We begin by observing that,
by definition, the functions in $V_\ell$ are homogeneous of degree $\ell$, so that, for every $v \in V_\ell$, every $x \in {\mathbb R}_\ast^d$ and every $r >0$, there holds the identity
$$
v ( r x ) = r^{\ell} v(x)\,.
$$
Differentiating once with respect to $r$ produces
$$
\sum_{j} x_j \frac{\partial v}{\partial x_j} (r x) = \ell \,r^{\ell-1} v(x) \,.
$$
Differentiating again, we obtain
$$
\sum_{i,j} x_i x_j \frac{\partial^2 v}{\partial x_i \partial x_j} ( r x) =  ( \ell-1) \ell  \,r^{\ell-2} v(x)\,.
$$
Evaluating at $r=1$, we readily deduce the 
\begin{lemma}
For every $v \in V_\ell$, the following identities hold:
$$
\sum_{j} x_j \frac{\partial v}{\partial x_j} = \ell \,v
$$
and
$$
\sum_{j} x_j^2 \frac{\partial^2 v}{\partial x_j^2} + 2 \sum_{i<j} x_i x_j \frac{\partial^2 v}{\partial x_i \partial x_j} = (\ell-1) \ell \,v\,.
$$
\label{casimirLemma}
\end{lemma}

For the proof of the proposition, we use the formulae 
(\ref{KinfinitesimalGenerators}-\ref{NinfinitesimalGenerators}) for the infinitesimal 
generators associated with the representation $T_\ell$. Equation (\ref{AinfinitesimalGenerators})
yields
\begin{multline}
\sum_{i<j} A_{ij}^2 = \sum_{i<j} \left ( x_i \frac{\partial}{\partial x_i} - x_j \frac{\partial}{\partial x_j} \right ) 
\left ( x_i \frac{\partial}{\partial x_i} - x_j \frac{\partial}{\partial x_j} \right ) \\
= \sum_{i<j} \left ( x_i \frac{\partial}{\partial x_i} + x_j \frac{\partial}{\partial x_j} - 2 x_i x_j \frac{\partial^2}{\partial x_i \partial x_j} + x_i^2 \frac{\partial^2}{\partial x_i^2} + x_j^2 \frac{\partial^2}{\partial x_j^2} \right ) \\
= (d-1) \ell + (d-1) \sum_{j} x_j^2 \frac{\partial}{\partial x_j^2} - 2 \sum_{i<j} x_i x_j \frac{\partial^2}{\partial x_i \partial x_j}
\label{sumOfAsquared}
\end{multline}
where, to obtain the last equality, we have made use of the first identity in Lemma \ref{casimirLemma}.
This, together with the second identity in the lemma, provides two equations
for the ``unknowns''
$$
\sum_j x_j^2 \frac{\partial^2}{\partial x_j^2}\;\;\text{and}\;\;
2 \sum_{i<j} x_i x_j \frac{\partial^2}{\partial x_i \partial x_j}\,.
$$
We deduce the identities
\begin{equation}
\sum_j x_j^2 \frac{\partial^2}{\partial x_j^2} = \frac{(\ell -d) \ell}{d} + \frac{1}{d} \sum_{i<j} A_{ij}^2
\label{formulaForX}
\end{equation}
and
\begin{equation}
2 \sum_{i<j} x_i x_j \frac{\partial^2}{\partial x_i \partial x_j} = \frac{d-1}{d} \ell^2 - \frac{1}{d} \sum_{i<j} A_{ij}^2\,.
\label{formulaForY}
\end{equation}

Now, Equation (\ref{NinfinitesimalGenerators}) yields
\begin{equation}
\sum_{i \ne j} N_{ij}^2 = \sum_{i \ne j} x_i^2 \frac{\partial^2}{\partial x_j^2}\,.
\label{sumOfNsquared}
\end{equation}
We then deduce from 
Equation (\ref{KinfinitesimalGenerators}) that
\begin{multline}
\Delta_K := \sum_{i<j} K_{ij}^2 = \sum_{i<j} \left ( x_i \frac{\partial}{\partial x_j} - x_j \frac{\partial}{\partial x_i} \right ) 
\left ( x_i \frac{\partial}{\partial x_j} - x_j \frac{\partial}{\partial x_i} \right ) \\
= \sum_{i<j} \left ( x_i^2 \frac{\partial^2}{\partial x_j^2} + x_j^2 \frac{\partial^2}{\partial x_i^2} 
- x_i \frac{\partial}{\partial x_i} - x_j \frac{\partial}{\partial x_j} - 2 x_i x_j \frac{\partial^2}{\partial x_i \partial x_j} \right ) \\
= \sum_{i \ne j} N_{ij}^2 - (d-1) \ell -2 \sum_{i<j}  x_i x_j \frac{\partial^2}{\partial x_i \partial x_j}
\label{sumOfKsquared}
\end{multline}
The desired result then follows from Equation (\ref{formulaForY}).

\subsection{The replica trick}
\label{quasiSolvableSubsection}
This ``trick" refers to the well-known fact that there are special values of $\ell$ for which some part of the spectrum of the transfer operator may be obtained by purely algebraic means; see for instance \cite{Va} or \cite{CPV}, \S 2.4.2. 
The explanation is that the infinitesimal generators map polynomials to polynomials. Hence the
\begin{proposition}
Let $\ell$ be an even natural number. Then the subspace ${\mathcal P}_\ell$ of $V_\ell$, consisting of the homogeneous polynomials of degree $\ell$, is invariant under the representation $T_\ell$, and hence also
under the transfer operator ${\mathscr T}_\ell$ defined by Equation (\ref{transferOperator}) and the operator $\Delta_K - k^2 \sum_{i<j} A_{ij}^2$ in Equation (\ref{perturbedSpectralProblem}). 
\label{quasiSolvableProposition}
\end{proposition}
It follows in particular that the restriction of the transfer operator to ${\mathcal P}_\ell$ may be represented by a matrix of order
$$
\text{\rm dim} \,{\mathcal P}_\ell = \binom{\ell+d-1}{d-1}\,.
$$
Furthermore, since (for $\ell$ an even natural number) the function $1_\ell$ belongs to ${\mathcal P}_\ell$, the relevant eigenvalue $\lambda (\ell)$ of the transfer operator must coincide with the leading root of the characteristic polynomial of that matrix.

This remains true in the continuum limit leading to the spectral problem (\ref{perturbedSpectralProblem}), where the object is to compute $\mu(k,\ell)$, and we shall find upon inspection that, when $\ell$ is an even natural number, there exist non-trivial subspaces of ${\mathcal P}_\ell$, of much lower dimension, that contain $1_\ell$ and are invariant under the operator $\Delta_K - k^2 \sum_{i < j} A_{ij}^2$.

\section{The Iwasawa realisation}
\label{iwasawaSection}
Even homogeneous functions of $d$ variables are completely determined by the values they take on (the ``upper half'' of)  $S^{d-1}$, the unit hypersphere in $d$-dimensional space. The elements $g \in \text{SL}(d,{\mathbb R})$ act on $S^{d-1}$ as follows: for every $\xi \in S^{d-1}$,
\begin{equation}
\xi \cdot g :=  \frac{\xi g}{\left | \xi g \right |}\,.
\label{actionOnTheSphere}
\end{equation}

Let $T_{\ell}$ denote the representation in the space $V_{\ell}$ introduced in \S \ref{transferOperatorSubsection}. 
We denote by $Q$ the operator that assigns to every $v \in V_{\ell}$ its restriction $\widetilde{v}$ to $S^{d-1}$.
We then set
$$
\widetilde{V}_{\ell} := \left \{ Qv :\, v \in V_\ell \right \}\,.
$$ 
In particular, from the definition (\ref{unitFunction}) of $1_{\ell}$, we see that
$$
\forall\,\xi \in S^{d-1}\,,\;\;\widetilde{1}_{\ell} (\xi) = 1\,.
$$
We obtain the Iwasawa realisation $\widetilde{T}_{\ell}$ of $T_{\ell}$ on the representation space $\widetilde{V}_{\ell}$ via
$$
\widetilde{T}_{\ell} \,Q = Q \,T_{\ell}\,.
$$
Explicitly,
\begin{multline}
\notag
\left [ \widetilde{T}_{\ell} (g) \,\widetilde{v} \right ] (\xi) = \left [ \widetilde{T_{\ell} (g) \,v} \right ]  (\xi) = \left [ T_{\ell} (g) \,v \right ] \left ( \xi  \right ) = v \left (  {\xi} g \right ) \\
=  \left | \xi g \right |^{\ell} v \left ( \xi \cdot g \right )  = \left | \xi g \right |^{\ell} \widetilde{v} \left ( \xi \cdot g \right ) 
\end{multline}
where, in the penultimate equality, 
we have used the fact that $v$ is even and homogeneous of degree $\ell$. We shall henceforth drop the tilde.

\subsection{The fluctuation theorem}
Let $d \xi$ denote a measure on the hypersphere that is invariant under the action of the subgroup $K$. It may be shown--- see for instance \cite{FT}--- that, for every $g \in \text{SL}(d,{\mathbb R})$, 
\begin{equation}
d ( \xi \cdot g ) = \left | \xi g \right |^{-d} \,d \xi\,.
\label{jacobianOfTheAction}
\end{equation}
Using the inner product
$$
\braket{f | v} := \int_{S^{d-1}} d \xi \,\overline{f(\xi)} \, v(\xi)
$$
where the bar denotes complex conjugation, we may write
\begin{multline}
\notag
\braket{f | T_\ell (g) v} = \int_{S^{d-1}} d \xi\,\overline{f(\xi)} \, \left | \xi g \right |^\ell v \left ( \xi \cdot g \right ) \overset{\underset{\downarrow}{\xi' = \xi \cdot g}}{=} \int_{S^{d-1}} d \xi'\,\left | \xi g \right |^{d+\ell} \overline{f \left ( \xi' \cdot g^{-1} \right )} \, v \left ( \xi' \right ) \\ 
= \int_{S^{d-1}} d \xi'\,\left | \xi' g^{-1} \right |^{-d-\ell} \overline{f \left ( \xi' \cdot g^{-1} \right )} \, v \left ( \xi' \right )\,.
\end{multline}
This gives an obvious formula for the formal adjoint of the operator $T_\ell (g)$. If we then consider a random $g$, and average over the group, we deduce that
the operator
\begin{equation}
{\mathscr T}_\ell^\dag (g) := {\mathbb E} \left [ T_{-\overline{\ell}-d} \left ( g^{-1} \right ) \right ]
\label{adjointOfT}
\end{equation}
defined on $V_{-\overline{\ell}-d}$
is the formal adjoint of the transfer operator ${\mathscr T}_\ell$ defined on $V_\ell$ in the previous paragraph.

Now, for our particular product, the random matrix $g$ has a very particular structure, displayed in Equation (\ref{jacobian}): its inverse is obtained by changing the sign of $\delta t$ and reversing the order in which the factors occur. In the continuum limit where we neglect terms of order $o(\tau^2)$,
neither the sign of $\delta t$ nor the order matters. Since the spectrum of the adjoint is the complex conjugate of the spectrum, we deduce the
\begin{proposition}
For $k = 1/\sqrt{d}$, there holds
\begin{equation}
\Lambda \left (k, \ell \right ) = \Lambda \left (k,-\ell-d \right )\;\;\text{and}\;\;\Upsilon (k,s) = \Upsilon (k,-s) - s\,d\,.
\label{spectralSymmetry}
\end{equation}
\label{fluctuationProposition}
\end{proposition}
The proposition and its proof
are essentially taken from Vanneste's Proposition 2 \cite{Va}.
It is evident from the formulae presented in \S \ref{outlineSubsection} that the symmetry property (\ref{spectralSymmetry}) actually holds for every $k$. It is a particular instance of the {\em fluctuation theorem} which is well-known in non-equilibrium statistical mechanics; see \cite{De}, \S 10 and \cite{TH}, \S 11.3.3. 

\subsection{Polyspherical angles}
Almost every point in $\xi \in S^{d-1}$ may be expressed in the form
\begin{align*}
\xi_1 =& \sin \theta_d \cdots \sin \theta_3 \sin \theta_2 \\
\xi_2 =& \sin \theta_d \cdots \sin \theta_3 \cos \theta_2 \\
\vdots& \\
\xi_{d-1} =& \sin \theta_d \cos \theta_{d-1} \\
\xi_d =& \cos \theta_d
\end{align*}
for some
$$
\theta := \begin{pmatrix} \theta_2 & \cdots & \theta_d \end{pmatrix} \in
\Theta := \left \{ \theta \in {\mathbb R}^{d-1} :\,-\pi < \theta_2 < \pi\;\text{and}\;0<\theta_3 ,\ldots, \theta_d <\pi \right \}\,.
$$
The angles on the right-hand side are called {\em polyspherical angles} (see \cite{Vi}, Chapter IX).
Conversely, given the left-hand side, the polyspherical angles may be determined via the formulae
\begin{equation}
\cos \theta_j = \frac{\xi_j}{\sqrt{\xi_1^2 + \cdots \xi_j^2}}\;\;\text{and}\;\;\sin \theta_j = s_j \frac{\sqrt{\xi_1^2 + \cdots \xi_{j-1}^2}}{\sqrt{\xi_1^2 + \cdots \xi_j^2}}\,,\;\;1 < j \le d\,.
\label{polysphericalAngles}
\end{equation}
Here $s_j = 1$ for $j >2$ whilst $s_2$ may take the value $1$ or $-1$, depending on $\xi$. In what follows, we shall  think of functions in the
representation space as functions on $\Theta$.

\subsection{Calculation of the infinitesimal generators}
Let now $X \in {\mathfrak sl} (d)$. In polyspherical coordinates, the infinitesimal generator--- temporarily denoted $X'$ to distinguish it from the matrix $X$---
of the one-parameter subgroup
$$
g(t) = e^{t X} \in \text{SL}(d,{\mathbb R})
$$
takes the form
\begin{equation}
X' = \ell\, \dot{r} + \sum_{j=2}^d \dot{\theta}_{j} \frac{\partial}{\partial \theta_{j}}
\label{infinitesimalGeneratorRealisedOnTheStiefelManifold}
\end{equation}
for some $\theta$-dependent coefficients $\dot{r}$ and $\dot{\theta}_{j}$.  
To compute these coefficients, put
$$
\xi e^{tX} = r(t) \,\xi (t)\,,\;\;r(t) > 0\,,\;\;\xi (t) \in S^{d-1}\,.
$$
Differentiation with respect to $t$ yields, after setting $t=0$,
\begin{equation}
\xi X = \dot{r} \xi + \dot{\xi}
\label{firstInfinitesimal}
\end{equation}
where we have used the fact that $r(0)=1$ and $\xi(0) = \xi$. For simplicity, we have also written $\dot{r}$ and $\dot{\xi}$ instead of  $\dot{r}(0)$ and $\dot{\xi}(0)$. Denoting by $(\cdot,\cdot)$ the usual dot product in ${\mathbb R}^d$, we remark that
$$
1 = \left ( \xi (t), \xi (t) \right ) \implies \dot{\xi} \perp \xi\,. 
$$
Hence
\begin{equation}
\dot{r} = \left ( \xi X, \xi \right )
\label{formulaForRdot}
\end{equation}
and Equation (\ref{firstInfinitesimal}) becomes
\begin{equation}
\dot{\xi} = \xi X - \left ( \xi X, \xi \right ) \xi \,.
\label{secondInfinitesimal}
\end{equation}
Differentiation of the formula for $\cos \theta_j$ in Equation (\ref{polysphericalAngles}) then yields the following formula for the $\dot{\theta}_j$ in terms of the $\dot{\xi}_j$:
\begin{equation}
\dot{\theta}_j = \frac{\xi_1 \left ( \dot{\xi}_1 \xi_j - \xi_1 \dot{\xi}_j \right )+\cdots+\xi_{j-1} \left ( \dot{\xi}_{j-1} \xi_j - \xi_{j-1} \dot{\xi}_j \right )}{\sin \theta_j  \left ( \xi_1^2 + \cdots + \xi_{j}^2\right )^{\frac{3}{2}}} \,,\;\;1 < j \le d\,.
\label{formulaForDotTheta}
\end{equation}

\begin{example}
The infinitesimal generators, in the case $d=2$, are shown in the second column of Table \ref{2dInfinitesimalGeneratorTable}.
\label{iwasawaGeneratorsForKisOneAndDisTwoExample}
\end{example}

\begin{table}
\begin{tabular}{c c c}
& Iwasawa & Gauss \\
\hline
\hline
& & \\
$K_{12}$  &  $-\frac{d}{d \theta}$ & $-\ell z + (z^2+1) \frac{d}{dz}$  \\
& & \\
$A_{12}$ &   $-\ell \cos (2\theta) +\sin (2 \theta) \frac{d}{d \theta}$ & $\ell - 2 z \frac{d}{dz}$  \\
& & \\
 ${N}_{12}$ &  $(\ell/2) \sin (2 \theta) - (\sin \theta)^2 \frac{d}{d \theta}$  & $\frac{d}{d z}$   \\
 & & \\
 ${N}_{21}$  & $(\ell/2) \sin (2 \theta) + (\cos \theta)^2 \frac{d}{d \theta}$   & $\ell z -z^2 \frac{d}{d z}$  \\
 & & \\
 \hline
\\[0.125cm]
\end{tabular}
\caption{The case $d=2$: infinitesimal generators associated with the representation $T_\ell$ for various one-parameter subgroups.
The first column corresponds to the Iwasawa realisation and the second to the Gauss realisation of the representation. For simplicity, we have dropped the subscripts of $\theta$.}
\label{2dInfinitesimalGeneratorTable}
\end{table}

\begin{example}
For the case $d=3$, the infinitesimal generators are as follows:
$$
K_{12} = - \frac{\partial}{\partial \theta_2}
$$
$$
K_{13} = - \cos \theta_2 \cot \theta_3 \frac{\partial}{\partial \theta_2} -\sin \theta_2 \frac{\partial}{\partial \theta_3}
$$
$$
K_{23} = \sin \theta_2 \cot \theta_3 \frac{\partial}{\partial \theta_2} -\cos \theta_2 \frac{\partial}{\partial \theta_3}
$$
\newline
$$
N_{12} = \frac{\ell}{2} \sin (2 \theta_2)  (\sin \theta_3)^2-(\sin \theta_2)^2 \frac{\partial}{\partial \theta_2}
+ \frac{1}{4} \sin ( 2 \theta_2 ) \sin ( 2 \theta_3 ) \frac{\partial}{\partial \theta_3}
$$
$$
N_{13} = \frac{\ell}{2} \sin \theta_2 \sin (2 \theta_3) - \sin \theta_2 (\sin \theta_3)^2 \frac{\partial}{\partial \theta_3}
$$
$$
N_{23} = \frac{\ell}{2} \cos \theta_2 \sin (2 \theta_3) - \cos \theta_2 (\sin \theta_3)^2 \frac{\partial}{\partial \theta_3}
$$
\newline
$$
N_{21} = \frac{\ell}{2} \sin (2 \theta_2) (\sin \theta_3 )^2 + (\cos \theta_2 )^2 \frac{\partial}{\partial \theta_2}
+ \frac{1}{4} \sin (2 \theta_2 ) \sin (2 \theta_3) \frac{\partial}{\partial \theta_3}
$$
$$
N_{31} = \frac{\ell}{2} \sin \theta_2 \sin ( 2 \theta_3) + \cos \theta_2 \cot \theta_3 \frac{\partial}{\partial \theta_2} + \sin \theta_2 ( \cos \theta_3 )^2 \frac{\partial}{\partial \theta_3}
$$
$$
N_{32} = \frac{\ell}{2} \cos \theta_2 \sin ( 2 \theta_3) - \sin \theta_2 \cot \theta_3 \frac{\partial}{\partial \theta_2} + \cos \theta_2 ( \cos \theta_3 )^2 \frac{\partial}{\partial \theta_3}\,.
$$

Formulae for the infinitesimal generators $A_{ij}$ will be given in \S \ref{continuumLimit3dSection}.
\label{iwasawaGeneratorsForKisOneAndDisThreeExample}
\end{example}

\section{Two dimensions}
\label{continuumLimit2dSection}
In this section, we confine ourselves to the case $d=2$ and compute the dominant eigenvalue of the transfer operator by solving the companion eigenvalue problem (\ref{perturbedSpectralProblem}).
In the Iwasawa realisation, described in \S \ref{iwasawaSection}, the functions in the representation space are defined on the unit circle, parametrised by means of the angle $\theta$. From Table \ref{2dInfinitesimalGeneratorTable}, we see that Equation (\ref{perturbedSpectralProblem})
takes the concrete form
\begin{multline}
\mu v = \left [ 1- k^2 \sin^2 (2 \theta) \right ] \frac{d^2 v}{d \theta^2} \\ + (\ell-1) \,k^2 \sin ( 4 \theta) \frac{d v}{d \theta} 
 - k^2 \ell \left [ 2 \sin^2 ( 2 \theta) + \ell \cos^2 (2 \theta) \right ] v\,.
\label{perturbedSpectralProblemForDEqual2}
\end{multline}

It is clear that the differential operator on the right-hand side of Equation (\ref{perturbedSpectralProblemForDEqual2}) maps the space of even $(\pi/2)$-periodic functions to itself. Now, as we have seen, in the Iwasawa realisation the function $1_\ell$ is identically equal to $1$, and thus even and $(\pi/2)$-periodic.
Therefore, our eigenvalue problem is: Find $\mu$ such that Equation (\ref{perturbedSpectralProblemForDEqual2}) admits even $(\pi/2)$-periodic solutions. It will be convenient to normalise the corresponding eigenfunction by imposing the condition
\begin{equation}
v(0) =1\,.
\label{normalisationCondition}
\end{equation}
This eigenvalue problem is solvable when $\ell=0$ or when $k^2=0$. We can thus expand the solution in ascending powers of $\ell$ for $k^2$ fixed, or in ascending powers of $k^2$ for $\ell$ fixed.

\subsection{Calculation of $\Lambda (k,\ell)$ to $o(\ell^2)$} 
\label{smallEllSubsection}
For this purpose, it is convenient to introduce the new independent variable
$$
u = \int_0^\theta \frac{dt}{\sqrt{1-k^2 \sin^2(2t)}}
$$
and the new unknown $w$, related to $v$ via
$$
v(\theta) = w \left ( \int_0^\theta \frac{dt}{\sqrt{1-k^2 \sin^2(2t)}} \right )\,,
$$
so that Equation (\ref{perturbedSpectralProblemForDEqual2}) becomes
\begin{equation}
\mu w = w'' + 2 \ell k^2 \frac{\text{\rm sn}(2u) \,\text{\rm cn} (2u)}{\text{\rm dn} (2 u)} w'
- \ell k^2 \left [ 2 \,\text{\rm sn}^2 (2u) + \ell \,\text{\rm cn}^2(2u)\right ] w\,.
\label{spectralProblemInTermsOfW}
\end{equation}
Here, the prime denotes differentiation with respect to $u$, and $\text{\rm cn}$, $\text{\rm sn}$ and $\text{\rm dn}$ are the Jacobian elliptic functions; see \cite{DLMF}, Chapter 22, or \cite{La}, Chapter 2. It is easily verified that the requirement that $v$ be an even $(\pi/2)$-periodic function of $\theta$ translates into the requirement that $w$ be an even $\boldsymbol{\mathsf{K}}(k)$-periodic  function of $u$, where
\begin{equation}
\boldsymbol{\mathsf{K}}(k) := \int_{0}^{\frac{\pi}{2}} \frac{dt}{\sqrt{1-k^2 \sin^2 t}}
\label{ellipticK}
\end{equation}
is the complete elliptic integral of the first kind; see \cite{La}, \S 3.8.
We will often omit to indicate explicitly the dependence of $\boldsymbol{\mathsf{K}}$ on $k$.
The normalisation condition (\ref{normalisationCondition}) translates into
$$
w(0)=1\,.
$$

Let us now look for a solution of the form
\begin{equation}
\mu = \sum_{i=0}^\infty   \mu_{i} \,\ell^i\,,\;\;w(u) = \sum_{i=0}^\infty   w_i (u) \,\ell^i\,. 
\label{expansionInPowersOfL}
\end{equation}
Then
\begin{equation}
w_0''-\mu_0 \,w_0 =0\,,
\label{zerothOrder}
\end{equation}
\begin{equation}
w_1'' -\mu_0 \,w_1  =  \mu_1 \,w_0- 2 k^2 \left [ \frac{\text{\rm sn}(2u) \,\text{\rm cn} (2u)}{\text{\rm dn} (2 u)} \,w_0' -  \text{\rm sn}^2 (2 u) \,w_0 \right ]\,,
\label{firstOrder}
\end{equation}
and, for $i = 2\,,3,\,\ldots$,
\begin{multline}
 w_i'' - \mu_0 w_i = \sum_{j=1}^i \mu_j \,w_{i-j}  \\ - k^2 \left [  \frac{2 \,\text{\rm sn}(2u) \,\text{\rm cn} (2u)}{\text{\rm dn} (2 u)} \,w_{i-1}' 
- 2 \,\text{\rm sn}^2 (2 u) \,w_{i-1} -  \text{\rm cn}^2 (2 u) \,w_{i-2} \right ]\,.
\label{higherOrder}
\end{multline}
We deduce
\begin{equation}
\mu_0 = \mu_{0,n} := - \left ( \frac{2 \pi n}{\boldsymbol{\mathsf{K}}} \right )^2\,,\;\;w_0 (u) = w_{0,n}(u) := \cos \left (  \frac{2 \pi n}{\boldsymbol{\mathsf{K}}} u \right )\,,\;\;n \in {\mathbb N}\,.
\label{lowestOrderTerm}
\end{equation}
Bearing in mind Equation (\ref{generalisedLyapunovExpansion}), we see that the leading eigenvalue of the transfer operator is obtained by choosing $n=0$, the next largest is obtained by choosing $n=1$, and so on.

Now, fix $n \in {\mathbb N}$.
Let us indicate briefly how, in principle, higher-order terms may be computed by recurrence on $i$. The central point is that, for every $i$, $w_i$--- as well as the right-hand sides of 
Equations (\ref{firstOrder}-\ref{higherOrder})--- may be expressed as an infinite linear combination of the $w_{0,j}$, $j \in {\mathbb N}$. Our task is thus to determine these coefficients, as well as the value of $\mu_i$, knowing $w_l$ and $\mu_l$ for $l<i$.
For the equation satisfied by $w_i$ to be consistent, we must select the parameter $\mu_{i}$ so that the coefficient multiplying $w_{0,n}$ in the expansion of the right-hand side vanishes. Once $\mu_{i}$ has been determined, one can, for $j \ne n$, work out the coefficient
multiplying  $w_{0,j}$ in the expansion of $w_i$ by equating the coefficients on both sides. The only remaining coefficient in the expansion, namely that multiplying $w_{0,n}$, can be determined by imposing (for $i > 0$) the condition $w_i (0)=0$.

Let us give a partial illustration of this procedure for the case $i=1$.
Replacing $\mu_0$ by $\mu_{0,n}$
and $w_0$ by $w_{0,n}$ respectively in Equation (\ref{firstOrder}), we obtain, after some re-arrangement,
\begin{equation}
w_{1}'' - \mu_{0,n} \,w_{1} = \mu_{1}\,w_{0,n} + 2 k^2 \left [ \text{\rm sn}^2 (2u) \,w_{0,n} - \frac{\text{\rm sn}(2u) \,\text{\rm cn} (2u)}{\text{\rm dn} (2 u)} \,w_{0,n}' \right ]\,.
\label{firstOrderEquation}
\end{equation}
The right-hand side has an expansion of the form
$$
\sum_{j=0}^\infty r_j \cos \left (  \frac{2 \pi j}{\boldsymbol{\mathsf{K}}} u \right )\,.
$$
The consistency requirement $r_n=0$ produces
\begin{multline}
\mu_{1} \int_0^{\boldsymbol{\mathsf{K}}} d u \,w_{0,n}^2(u) \\
= -2 k^2 \int_0^{\boldsymbol{\mathsf{K}}} du \, w_{0,n} (u) \left [ \text{\rm sn}^2 (2u) \,w_{0,n} (u)- \frac{\text{\rm sn}(2u) \,\text{\rm cn} (2u)}{\text{\rm dn} (2 u)} \,w_{0,n}' (u)\right ]\,.
\end{multline}
Then, by using the formulae (2.16) and (2.23) from Chapter II of Oberhettinger's collection of Fourier expansions \cite{Ob}, we eventually deduce
\begin{equation}
\mu_1 = \mu_{1,n} := 2 \left ( \frac{\boldsymbol{\mathsf{E}}}{\boldsymbol{\mathsf{K}}}-1 \right ) + \left ( \frac{2 \pi}{\boldsymbol{\mathsf{K}}} \right )^2 \frac{n q^{2n}}{1-q^{4n}}\,,\;\; n \in {\mathbb N}\,,
\label{firstOrderEigenvalueCorrection}
\end{equation}
where 
\begin{equation}
\boldsymbol{\mathsf{E}}(k) := \int_{0}^{\frac{\pi}{2}} dt\,\sqrt{1-k^2 \sin^2 t}
\label{ellipticE}
\end{equation}
is the complete elliptic integral of the second kind--- see \cite{La}, \S 3.8--- and
\begin{equation}
\label{nome}
q  := \exp \left [ -\pi \frac{\boldsymbol{\mathsf{K}}(k')}{\boldsymbol{\mathsf{K}}(k)} \right ]
\end{equation}
is the so-called {\em nome}, with $k' := \sqrt{1-k^2}$. 

The calculation of $w_1$, and of higher-order terms for general $n$ is impractical. However, in the special case of interest--- namely $n=0$--- we obtain
\begin{equation}
w_{1} (u) = - \int_0^u dt \,\text{\rm zn} ( 2t)
\label{firstOrderEigenfunctionCorrection}
\end{equation}
where $\text{\rm zn}$ is Jacobi's zeta function; see \cite{La}, \S 3.6, and Formula (2.25) in \cite{Ob}. It then follows that, for $n=0$,
\begin{equation}
\mu_2 = \mu_{2,0} := 1-k^2 - \frac{\boldsymbol{\mathsf{E}}}{\boldsymbol{\mathsf{K}}} + \frac{2 \pi^2}{\boldsymbol{\mathsf{K}}^2} \sum_{j=1}^\infty \left ( \frac{q^{2j}}{1-q^{4j}} \right )^2 - \frac{6 \pi^2}{\boldsymbol{\mathsf{K}}^2} \sum_{j=1}^\infty \left ( \frac{q^{2j-1}}{1-q^{4j-2}} \right )^2\,.
\label{secondOrderEigenvalueCorrection}
\end{equation}

\begin{remark}
Elliptic functions are widely used in analytic number theory, and this leads to interesting alternative ways of expressing this last formula. For instance, if we introduce the divisor function (see \cite{DLMF}, \S 27.2)
$$
\sigma_\alpha (j) := \sum_{n | j} n^\alpha
$$
where the sum is over the divisors of the natural number $j$, then it may be shown that
$$
\sum_{j=1}^\infty \left ( \frac{q^{2j}}{1-q^{4j}} \right )^2 = \sum_{j=1}^\infty \sigma_1 (j) \,q^{4j}
$$
and
$$
\sum_{j=1}^\infty \left ( \frac{q^{2j-1}}{1-q^{4j-2}}  \right )^2 = \sum_{j=1}^\infty \sigma_1 (j) \,q^{2j} - \sum_{j=1}^\infty \sigma_1 (j) \,q^{4j}\,.
$$
The sum of the power series appearing on the right-hand side may be expressed in terms of the Eisenstein series $G_2$; see Apostol \cite{Ap}, \S 3.10.
\label{divisorRemark}
\end{remark}

\subsection{Expansion in powers of $k^2$}
\label{kexpansionSubsection}
As an alternative, let us look for solutions of Equation (\ref{perturbedSpectralProblemForDEqual2}) of the form
\begin{equation}
v(\theta) = \sum_{j=0}^\infty v_j \cos (4 j \theta)
\label{cosineSeries}
\end{equation}
so that the normalisation condition (\ref{normalisationCondition}) becomes
\begin{equation}
\sum_{j=0}^\infty v_j = 1\,.
\label{recurrenceNormalisationCondition}
\end{equation}
Substitution into the equation yields the following recurrence relation for the coefficients $v_j$:
\begin{equation}
k^2 a_j v_{j-1} + \left ( \mu+ 16 j^2 + k^2 b_j \right ) v_j + k^2 c_j v_{j+1} = 0\,,\;\;j \in {\mathbb N}\,,
\label{cosineRecurrenceRelation}
\end{equation}
where
\begin{equation}
a_j = \begin{cases} 
0 & \text{if $j=0$} \\
 (\ell/2-1) \ell & \text{if $j=1$} \\
( \ell/2+1-2 j )(\ell/2+2 -2j) & \text{otherwise}
\end{cases}
\label{matrixCoefficientA}
\end{equation}
\begin{equation}
b_j = \ell^2/2 +  \ell - 8 j^2 
\label{matrixCoefficientB}
\end{equation}
and
\begin{equation}
c_j =  ( \ell/2 + 1+ 2 j )(\ell/2+2 + 2j)\,.
\label{matrixCoefficientC}
\end{equation}

To proceed, we look for a solution of the recurrence relation of the form
\begin{equation}
\mu = \sum_{n=0}^\infty  \mu^{(n)} k^{2n}  \,,\;\;v_j = \sum_{n=0}^\infty  v_{j}^{(n)} k^{2n} \,.
\label{expansionInPowersOfkSquare}
\end{equation}
Obviously,
\begin{equation}
\left ( \mu^{(0)} + 16j^2 \right ) v_j^{(0)} = 0\,,\;\; j \in {\mathbb N}\,.
\label{kFirstRecurrenceRelation}
\end{equation}
Therefore, taking account of the normalisation condition (\ref{recurrenceNormalisationCondition}), we deduce
\begin{equation}
\mu^{(0)} = - (4l)^2\;\;\text{and}\;\; v_{j}^{(0)} = \delta_{lj}\,,\;\;j \in {\mathbb N}\,,
\label{zerothOrderInk}
\end{equation}
where, for every choice of $l \in {\mathbb N}$, $\mu^{(0)}$ is the limit, as $k^2 \rightarrow 0$, of a particular $k^2$-dependent eigenvalue branch $\mu$.
Then, for $n \ge 0$, the recurrence relation
\begin{equation}
16 \left (l^2-j^2\right ) v_{j}^{(n+1)} = a_j v_{j-1}^{(n)} + b_j v_{j}^{(n)} + c_j v_{j+1}^{(n)} + \mu^{(n+1)} v_j^{(0)} + \sum_{r=1}^{n} \mu^{(r)} v_{j}^{(n+1-r)}
\label{kSecondRecurrenceRelation}
\end{equation}
provides the successive correction terms for $l$ fixed. For consistency, the right-hand side of Equation (\ref{kSecondRecurrenceRelation}) must vanish when $j=l$. This determines the eigenvalue coefficient $\mu^{(n+1)}$, and hence
the $v_{j}^{(n+1)}$ for $j \ne l$. The normalisation condition then imposes
$$
v_{l}^{(n+1)} = - \sum_{j \ne l} v_{j}^{(n+1)}\,.
$$
There are at most $l+n+1$ non-zero terms in the sum on the right-hand side, so that, at each order, the calculation is completely explicit. 

For $l=0$--- the case of particular interest--- we find
\begin{multline}
\mu = \mu_0 := -\frac{\ell \left(\ell +2\right)}{2} k^2 \\ + \frac{\left(\ell -2\right) \ell \left(\ell +2\right) \left(\ell +4\right)}{128} k^4 
+ \frac{\left(\ell -2\right) \ell  \left(\ell +2\right) \left(\ell +4\right)}{256} k^6 \\
-\frac{\left(\ell -2\right) \ell \left(\ell +2\right)  \left(\ell +4\right) \left(7 \ell^{4}+28 \ell^{3}+36 \ell^{2}+16 \ell -5248\right)}{2097152} k^8 + \cdots
\label{muExpansion}
\end{multline}

This leads to the expansions (\ref{expansionForCGF2}-\ref{expansionForRateFunction2}). By re-arranging the terms in the former series and collecting them in powers of $\ell$, we deduce expansions for {\em all} the cumulants:
\begin{equation}
\gamma_1 = 1-k^{2}-\frac{1}{8} k^{4}-\frac{1}{16} k^{6}-\frac{41}{1024} k^{8}-\frac{59}{2048} k^{10} + \cdots
\label{firstCumulantExpansion}
\end{equation}
\begin{equation}
\frac{\gamma_2}{2} = \frac{1}{2} -\frac{1}{2} k^{2}-\frac{1}{32} k^{4}-\frac{1}{64} k^{6}-\frac{81}{8192} k^{8}-\frac{115}{16384} k^{10}+ \cdots
\label{secondCumulantExpansion}
\end{equation}
\begin{equation}
\frac{\gamma_3}{6} = \frac{1}{32} k^{4}+\frac{1}{64} k^{6}+\frac{169}{16384} k^{8}+\frac{251}{32768} k^{10} + \cdots
\label{thirdCumulantExpansion}
\end{equation}
\begin{equation}
\frac{\gamma_4}{24} = \frac{1}{128} k^{4}+\frac{1}{256} k^{6}+\frac{361}{131072} k^{8}+\frac{571}{262144} k^{10} + \cdots
\label{fourthCumulantExpansion}
\end{equation}
and so on. These expansions converge for $0 \le k^2 < 1$ and thus in particular, for $k^2=1/2$. 

We can also compute expansions for other eigenvalue branches. For instance, the branch corresponding to the next largest eigenvalue is obtained when $l=1$, and we find 
\begin{multline}
\mu = \mu_1 := -16 -\frac{\ell^2+2 \ell-16}{2} k^2 - \frac{5 \ell^{4}+20 \ell^{3}+44 \ell^{2}+48 \ell -1152}{768} k^4 \\ - \frac{5 \ell^{4}+20 \ell^{3}+44 \ell^{2}+48 \ell-1152}{1536} k^6 
+ \cdots
\notag
\end{multline}

\subsection{The replica trick}
\label{quasiSolvableSubsection2d}
For $d=2$, we find by inspection that if $\ell$ is an even natural number
then the finite-dimensional space
$$
\bigoplus_{j=0}^{[\ell/4]} \text{span} \left \{  \cos ( 4j \theta) \right \},
$$
where $[ \cdot ]$ denotes the integer part,
is invariant under the operator $A_{12}^2$. As explained in \S \ref{quasiSolvableSubsection}, it follows that the restriction of $d^2/d \theta^2 - k^2 A_{12}^2$ to this subspace may be represented by a finite matrix and, since $1_\ell$ belongs to that subspace, the leading eigenvalue, say $\mu_0$,  
of the matrix coincides with $\mu (k,\ell)$. 
In what follows, we write down the (monic version of the) characteristic polynomial for the first few special values of $\ell$; the leading root of its $k^2$-dependent characteristic polynomial admits an expansion in powers of $k^2$, 
and this will provide some check for the results obtained in the previous paragraph.

For $\ell=2$, the characteristic polynomial is $\mu+4 k^2$. Hence $\mu(k,2)=\mu_0=-4 k^2$.

For $\ell=4$, the characteristic polynomial is
$$
\mu^2 + 16 (1+k^2) \mu +192 k^2 = 0\,.
$$
The roots are
\begin{multline}
\notag
\mu(k,4) = \mu_0 = -8 k^2 - 8 + 8 \sqrt{k^4 - k^2 + 1} \\
= -12 k^{2}+3 k^{4}+\frac{3}{2} k^{6}+\frac{3}{16} k^{8}-\frac{15}{32} k^{10} + \cdots
\end{multline}
and
\begin{multline}
\notag
\mu_1 = -8 k^2 - 8 - 8 \sqrt{k^4 - k^2 + 1} \\
= -16-4 k^{2}-3 k^{4}-\frac{3}{2} k^{6}-\frac{3}{16} k^{8}+\frac{15}{32} k^{10} + \cdots
\end{multline}
These expansions converge for $k^2 < 1$.

For $\ell=6$, the characteristic polynomial is
$$
\mu^2 + 8 ( 5 k^2 + 2) \mu + 48 k^2 ( 3 k^2 + 8)\,.
$$
The roots are
\begin{multline}
\notag
\mu(k,6) = \mu_0 = -20 k^2 - 8 + 8 \sqrt{4 k^4 - k^2 + 1} \\
= -24 k^{2}+15 k^{4}+\frac{15}{2} k^{6}-\frac{165}{16} k^{8}-\frac{615}{32} k^{10} + \cdots
\end{multline}
and
\begin{multline}
\notag
\mu_1 =  -20 k^2 - 8 - 8 \sqrt{4 k^4 - k^2 + 1} \\
= -16-16 k^{2}-15 k^{4}-\frac{15}{2} k^{6}+\frac{165}{16} k^{8}+\frac{615}{32} k^{10} + \cdots
\end{multline}
These expansions converge for $k^2 < 1/2$. 

The reader will easily verify that these small-$k^2$ expansions do agree with the formulae presented in the previous paragraph. These exact results indicate that the radius of convergence decreases with $\ell$; in particular, for $\ell \ge 6$,
we expect that the expansions do not converge for $k^2=1/2$. The singularities on the circle of convergence occur at complex values of $k$ such that the ``spectral gap'' separating the dominant eigenvalue branch from the others
 shrinks to zero.
  
\subsection{The band-center anomaly for Anderson's model at zero energy}
For $k^2 = 1/2$, Formula (\ref{growthRateForDis2}) yields
\begin{equation}
\gamma_1 \left ( \frac{1}{\sqrt{2}}\right ) =  2 \frac{\boldsymbol{\mathsf{E}}\left (\frac{1}{\sqrt{2}} \right )}{\boldsymbol{\mathsf{K}}\left (\frac{1}{\sqrt{2}} \right )}-1 = 4 \left [ \frac{\Gamma(\frac{3}{4})}{\Gamma ( \frac{1}{4} )} \right ]^2 = \frac{1}{\pi G^2}
\label{gaussConstant}
\end{equation}
where $G$ denotes the Gauss lemniscate constant. The same (apart from a trivial factor) Lyapunov exponent comes out of the analysis of a one-dimensional Anderson model which
exhibits the Kappus--Wegner band-center anomaly at zero energy; see \cite{DG,KW, ST}. The formula also appears explicitly in \cite{CLTT}, \S 6.2, and \cite{RT} in connection with products whose terms are matrices of the particular form
\begin{equation}
\text{SL}(2,{\mathbb R}) \ni g_n = e^{\alpha_n K_{12}} \,e^{w_n A_{12}}
\label{supersymmetricMatrix}
\end{equation}
where $\alpha_n$ and $w_n$ are independent random variables with zero mean and respective variances $D_{\alpha \alpha}$ and $D_{ww}$. More precisely, Comtet {\em et al.} found that, in the limit as both variances tend to zero,
the Lyapunov exponent of the product equals
\begin{equation}
\Omega := \left ( D_{\alpha \alpha}+D_{ww}  \right ) \frac{\boldsymbol{\mathsf{E}}\left (k\right )}{\boldsymbol{\mathsf{K}}\left ( k\right )} - D_{\alpha \alpha}
\label{comtetFormula}
\end{equation}
where
\begin{equation}
\label{comtetModulus}
k^2 := \frac{D_{ww}}{D_{\alpha \alpha} + D_{ww}}\,.
\end{equation}
In their study of the one-dimensional Dirac equation with a random mass, Ramola \& Texier showed how to map the Anderson model at zero energy to a product of precisely this form; see \cite{RT}, \S 8.2. 
There only remains to clarify the relationship between the products studied in the present paper and those considered by Comtet {\em et al.} in their \S 6.2. 

To this end, let $k$ be defined by Formula (\ref{comtetModulus}), so that
\begin{equation}
\frac{k}{k'} := \frac{k}{\sqrt{1-k^2}} = \sqrt{\frac{D_{ww}}{D_{\alpha \alpha}}}\,.
\label{ellipticRatio}
\end{equation}
Denote by $\widetilde{L} (k,\ell)$ the generalised Lyapunov exponent associated with Comtet {\em et al.}'s product. In the continuum limit, it is given by the formula
$$
\widetilde{L} (k,\ell) = \widetilde{\tau}^2  \,\widetilde{\Lambda} (k,\ell) + o \left ( \widetilde{\tau}^2 \right ) \;\;\text{as $\widetilde{\tau} \rightarrow 0+$}
$$
where $\widetilde{\tau}^2 = D_{\alpha \alpha}/2$ and $\widetilde{\Lambda}(k,\ell)$ is the leading eigenvalue of the operator
$$
K_{12}^2 + \frac{D_{ww}}{D_{\alpha \alpha}} A_{12}^2 \,.
$$
Recalling Equations (\ref{perturbedSpectralProblem}) and (\ref{ellipticRatio}), we deduce the relationship
$$
\widetilde{\Lambda} (k,\ell) = \mu \left ( \text{\rm i} \frac{k}{k'}, \ell \right )
$$
where $\mu(\cdot ,\ell)$ was calculated in the foregoing subsections. By setting $n=0$ in Equation (\ref{firstOrderEigenvalueCorrection}), it follows in particular that
\begin{multline}
\notag
\widetilde{\tau}^2 \frac{\partial \widetilde{\Lambda}}{\partial \ell} (k,0) = \frac{D_{\alpha \alpha}}{2}\, 2 \left [ \frac{\boldsymbol{\mathsf{E}}\left (\text{\rm i} \frac{k}{k'}\right )}{\boldsymbol{\mathsf{K}}\left ( \text{\rm i} \frac{k}{k'} \right )} - 1 \right ] 
= D_{\alpha \alpha} \left [ \frac{1}{(k')^2} \frac{\boldsymbol{\mathsf{E}}\left (k \right )}{\boldsymbol{\mathsf{K}}\left ( k\right )} - 1 \right ] \\
= D_{\alpha \alpha} \left [ \frac{D_{\alpha \alpha} + D_{ww}}{D_{\alpha \alpha}} \frac{\boldsymbol{\mathsf{E}}\left (k \right )}{\boldsymbol{\mathsf{K}}\left ( k\right )} - 1 \right ] = \Omega
\end{multline}
where, to obtain the second equality, we have made use of Formula 19.7.2 in \cite{DLMF} :
$$
\boldsymbol{\mathsf{E}}\left (\text{\rm i} \frac{k}{k'}\right ) = \frac{1}{k'} \boldsymbol{\mathsf{E}}\left (k\right ) \;\;\text{and}\;\; \boldsymbol{\mathsf{K}}\left (\text{\rm i} \frac{k}{k'}\right ) = k' \boldsymbol{\mathsf{K}}\left (k\right )\,.
$$

Our calculations therefore translate into results pertaining to products involving matrices of the form (\ref{supersymmetricMatrix}).
It is interesting to note that the case of equal variances, relevant to the Kappus--Wegner band-center anomaly, lies precisely on the circle of convergence of the expansion presented in \S \ref{kexpansionSubsection}.

\subsection{Relationship to the Kraichnan model of \S \ref{kraichnanSubsection}}
As we have already indicated, formulae like Equations (\ref{growthRateForDis2}-\ref{varianceForDis2}) were found by Chetrite {\em et al.} for the Kraichnan model. Indeed, the analysis in their \S 4.5 leads to the following conclusion:
the rate function associated with the two-dimensional Kraichnan model may be expressed, in their notation, as the Legendre transform of
\begin{equation}
\nu \mapsto ( 2 \alpha + \beta + \gamma ) \nu (\nu+1)-2 (\alpha+\gamma) E_{\nu,0} 
\label{chetriteLyapunovExponent}
\end{equation}
where $E_{\nu,0}$ is the dominant eigenvalue of the operator
\begin{equation}
{\mathcal L}_\nu := 2 \alpha \left ( \sin \phi \frac{\partial}{\partial \phi}  - \nu \cos \phi \right )^2 + 2 \gamma \frac{\partial^2}{\partial \phi^2}
\label{chetriteOperator}
\end{equation}
acting on $\pi$-periodic functions. The substitution $\phi = 2 \theta$, $\nu = \ell/2$ brings this operator into the form (see Table \ref{2dInfinitesimalGeneratorTable})
$$
\frac{\gamma}{2} K_{12}^2 + \frac{\alpha}{2}  A_{12}^2\,.
$$
In their calculation, both $\alpha$ and $\gamma$ assume positive values. The Kraichnan model therefore maps directly to the product studied by \cite{CLTT}, \S 6.2, and hence to ours. By the calculation of the previous paragraph, we deduce
$$
E_{\nu,0} := \frac{\gamma}{2} \,\mu \left ( \text{\rm i} \sqrt{\frac{\alpha}{\gamma}}, \,2 \nu \right )\,.
$$
We remark that, in the incompressible case $\beta = 0$, Formula (\ref{anisotropyDegree}) for the degree of anisotropy of the Kraichnan model yields
$$
\kappa = \frac{\alpha}{\gamma}\,.
$$
Hence the case of ``maximal anisotropy'' lies on the circle of convergence of the expansion presented in \S \ref{kexpansionSubsection}.

\section{Three dimensions}
\label{continuumLimit3dSection} 
The infinitesimal generators associated with the subgroup $A$ take the form
$$
A_{12} = - \ell \,( \sin \theta )^2 \cos (2 \varphi) + \sin (2 \varphi) \frac{\partial}{\partial \varphi} - \frac{1}{2} \sin (2 \theta) \cos (2 \varphi)  \frac{\partial}{\partial \theta}
$$
$$
A_{13} = \ell \left [ \left ( \sin \theta \right )^2  \left ( \sin \varphi \right )^2 - \left ( \cos \theta \right )^2 \right ] + \frac{1}{2} \sin (2 \varphi) \frac{\partial}{\partial \varphi} + \frac{1+ \left ( \sin \varphi \right )^2}{2} \sin ( 2 \theta) \frac{\partial}{\partial \theta}
$$
and
$$
A_{23} = 
\ell \left [ \left ( \sin \theta \right )^2  \left ( \cos \varphi \right )^2 - \left ( \cos \theta \right )^2 \right ] - \frac{1}{2} \sin (2 \varphi) \frac{\partial}{\partial \varphi} + \frac{1+ \left ( \cos \varphi \right )^2}{2} \sin ( 2 \theta) \frac{\partial}{\partial \theta}
$$
where, for simplicity, we write $\varphi$ instead of $\theta_2$, and $\theta$ instead of $\theta_3$.
The Casimir operator is 
$$
\Delta_K := \sum_{i<j} K_{ij}^2 =  \frac{1}{\sin^2 \theta} \frac{\partial^2}{\partial \varphi^2} + \frac{1}{\sin \theta} \frac{\partial}{\partial \theta} \sin \theta \frac{\partial}{\partial \theta}\,.
$$
We recognise the angular part of the familiar Laplacian in three-dimensional space.
The spectral problem for this operator is solved in \cite{DLMF}, \S 14.30: the eigenvalues are the numbers
$$
- l (l+1)\,,\;\;l \in {\mathbb N}\,.
$$
The eigenspace corresponding to $-l (l+1)$ has dimension $2 l+1$ and is spanned by
the (unnormalised) spherical harmonics 
$$
Y_{l}^{m} (\theta,\varphi) := \cos(m \varphi) P_l^{m} (\cos \theta)\;\;\text{or} \;\;\sin(m \varphi) P_l^{m} (\cos \theta)\,,\;\; 0 \le m \le l\,.
$$
Here,
$$
P_l^{m} (x) := \left ( 1-x^2 \right )^{m/2} \frac{d^m P_l}{d x^m} (x)
$$ 
denotes the associated Legendre function,
where $P_l$ is the Legendre polynomial of degree $l$, normalised by $P_l (1)=1$.
The $Y_{l}^{m}$ form an orthogonal set for the inner product
\begin{equation}
\braket{f|v} := \frac{1}{4 \pi} \int_{-\pi}^\pi d \varphi \int_0^\pi d \theta \sin \theta \,\overline{f(\theta,\varphi)}
v (\theta,\varphi)\,.
\label{sphereInnerProduct}
\end{equation}
Furthermore, every square integrable function on the unit sphere--- and hence the elements of $V_\ell$--- may be expanded in terms of these spherical harmonics. 

Set
\begin{equation}
e_{lm} (\theta,\varphi) := P_{2 l}^{2m} (\cos \theta) \cos(2 m \varphi)\,.
\label{basisFunctionsForDis3}
\end{equation}
From the explicit form of the $A_{ij}$, we immediately deduce the
\begin{lemma}
The subspace 
$$
\text{\rm span} \left \{ e_{lm} :\;\; l \in {\mathbb N}\,,\;\;0 \le m \le l \right \}
$$
of $V_\ell$
is invariant under the action of the $A_{ij}$, $1 \le i < j \le 3$.
\label{invarianceLemma}
\end{lemma}

The action of the operator 
\begin{equation}
{\mathscr A} := \sum_{i<j} A_{ij}^2
\label{Aoperator}
\end{equation}
on the basis functions may be deduced from the formulae in Appendix \ref{3dAppendix}:
\begin{equation}
{\mathscr A} e_{ij} = \sum_{p=-2}^2 \sum_{q=-2}^2 a_{pq} (i,j) \,e_{p+i,q+j}
\label{actionOfA}
\end{equation}
where the variable coefficients $a_{pq}(i,j)$ may be computed explicitly for any specific values of $p,q,i,j$.  An inspection of these coefficients reveals that the subspaces
$$
\text{span} \left \{ e_{lm} : l,m \in {\mathbb N}\,, \text{$m$ even}, 0 \le m \le l \right \}
$$
and
$$
\text{span} \left \{ e_{lm} : l,m \in {\mathbb N}\,, \text{$m$ odd}, 0 \le m \le l \right \}
$$
are invariant under the action of ${\mathscr A}$. Since $1_\ell = e_{00}$ belongs to the former subspace,
we shall look for a solution of Equation (\ref{perturbedSpectralProblem}) of the form
$$
v(\theta,\varphi) = \sum_{i=0}^\infty \sum_{j=0}^{\left [ i/2 \right ]} v_{ij} \,\varepsilon_{ij}(\theta,\varphi)
$$
where
\begin{equation}
\varepsilon_{ij}  := e_{i,2j} = P_{2i}^{4j} ( \cos \theta) \cos(4j \varphi)\,.
\label{3dBasisFunctions}
\end{equation}
Equation (\ref{actionOfA}) then produces
$$
{\mathscr A} \varepsilon_{ij} = \sum_{p=-2}^2 \sum_{q=-1}^1 a_{p,2q} (i,2j) \,\varepsilon_{p+i,q+j}\,.
$$
After substituting and equating the coefficients, we obtain the following difference equation for the unknown $v_{ij}$:
\begin{equation}
\mu \,v_{ij} = - 2i (2 i+1) \,v_{ij} - k^2 \left ( {\mathscr A}v \right )_{ij}\,,\;\;i,\,j \in {\mathbb N}\,,\;\;0 \le j \le \left [ i/2 \right ]\,,
\label{3dDifferenceEquation}
\end{equation}
where $\left ( {\mathscr A}v \right )_{ij}$ is defined implicitly by
$$
{\mathscr A} v = \sum_{i=0}^\infty \sum_{j=0}^{\left [ i/2 \right ]} \left ( {\mathscr A}v \right )_{ij} \varepsilon_{ij}\,.
$$

\subsection{The small-$k^2$ limit}
When $k=0$, the eigenvalues are given by $\mu = -2l (2l+1)$, $l \in {\mathbb N}$; the corresponding eigenspace is spanned by the $\varepsilon_{lm}$, $0 \le m \le \left [ l/2 \right ]$, so that the algebraic multiplicity
is $\left [l/2\right ]+1$.
As explained in the introduction, the eigenvalue of interest is the one that maximises the right-hand side of Equation (\ref{generalisedLyapunovExpansion}); for small $k^2$, the relevant branch therefore corresponds to $l=0$, the next most relevant branch corresponds to $l=1$, and so on.
For $l \in \{0,1\}$, the eigenspace is one-dimensional, so we can proceed as in the case $d=2$, except that it will be more convenient to normalise the eigenfunction corresponding to $l$ by imposing the condition
\begin{equation}
v_{l0} = 1\,.
\label{3dnormalisationCondition}
\end{equation}
We look for a solution of Equation (\ref{3dDifferenceEquation}) of the form
\begin{equation}
v_{ij} = \sum_{n=0}^\infty v_{ij}^{(n)} k^{2n}\,,\;\;\mu = \sum_{n=0}^\infty \mu^{(n)} k^{2n}\,.
\label{integralPowerSeries}
\end{equation}
Obviously,
$$
\mu^{(0)} = -2l (2l+1)\;\;\text{and}\;\;v_{ij}^{(0)} = \delta_{il} \delta_{j0}\,.
$$
Equations (\ref{3dDifferenceEquation}-\ref{integralPowerSeries}) then imply, for $n=0,\,1,\,\ldots$,
\begin{multline}
\label{3drecurrenceRelation}
\left [2l(2l+1) - 2i (2i+1) \right ] v_{ij}^{(n+1)} \\
= \mu^{(n+1)} v_{ij}^{(0)} + \left ( {\mathscr A} v^{(n)} \right )_{ij} + \sum_{r=1}^n \mu^{(r)} v_{ij}^{(n+1-r)}\,.
\end{multline}
Now, the normalisation condition (\ref{3dnormalisationCondition}) forces
\begin{equation}
\notag
v_{l0}^{(n+1)} = 0\,.
\end{equation}
Therefore, by setting  $i=l$, we deduce from Equation (\ref{3drecurrenceRelation}) that
\begin{equation}
\notag
\mu^{(n+1)} = - \left ( {\mathscr A} v^{(n)} \right )_{l0}\,.
\end{equation}
The $v_{ij}^{(n+1)}$, $i \ne l$, are then easily computed.

For $l=0$, this algorithm yields 
\begin{multline}
\mu = \mu_0 := -\frac{4 \ell \left( \ell +3\right)}{5} k^2 + \frac{12 \ell \left(\ell -2\right)  \left(\ell +3\right) \left(\ell +5\right)}{875} k^4 \\
-\frac{72 \left(\ell -2\right) \ell \left(\ell +3\right) \left(\ell +5\right)  \left(9 \ell^{2}+27 \ell -395\right)}{3128125} k^6 + \cdots
\label{3dexpansionForMU0}
\end{multline}
leading to the formula (\ref{expansionForCGF3}).  By collecting the terms in powers of $\ell$, we deduce the following expansions for the cumulants:
\begin{multline}
\gamma_1 = 2-\frac{12}{5} k^{2}-\frac{72}{175} k^{4}-\frac{34128}{125125} k^{6} \\
-\frac{17244576}{74449375} k^{8} -\frac{223736256}{1010384375} k^{10} + \cdots
\label{firstCumulantExpansion3}
\end{multline}
\begin{multline}
\frac{\gamma_2}{2} = \frac{2}{3}-\frac{4}{5} k^{2}-\frac{12}{875} k^{4}+\frac{5976}{625625} k^{6} \\
+\frac{59358528}{2605728125} k^{8}+\frac{22326336}{642971875} k^{10} + \cdots
\label{secondCumulantExpansion3}
\end{multline}
\begin{equation}
\frac{\gamma_3}{6} = \frac{72}{875} k^{4}+\frac{27432}{446875} k^{6}+\frac{30362256}{521145625} k^{8}+\frac{15082522656}{247544171875} k^{10} + \cdots
\label{thirdCumulantExpansion3}
\end{equation}
\begin{equation}
\frac{\gamma_4}{24} = \frac{12}{875} k^{4}+\frac{1584}{284375} k^{6}+\frac{32383188}{13028640625} k^{8}+\frac{60775592}{1237720859375} k^{10} + \cdots
\label{fourthCumulantExpansion3}
\end{equation}
and so on. Our calculations suggest that these expansions converge for $k^2 < R$, where $R$ is approximately $0.7$. 

For $l=1$, we obtain
\begin{multline}
\mu = \mu_1 := -6 -\frac{8 \left(\ell^2 +3 \ell-3\right)}{7} k^2 + \frac{96 \left(\ell -2\right) \left(\ell +5\right) \left(2 \ell +3\right)^{2}}{41503} k^4 \\
+\frac{128  \left(\ell -2\right) \left(\ell +5\right) \left(2 \ell +3\right)^{2} \left(8 \ell^{2}+24 \ell +361\right)}{26437411} k^6 + \cdots
\label{3dexpansionForMU1}
\end{multline}

For $l>1$, a new situation arises because there are--- not just one--- but  rather $[l/2]+1$ associated eigenfunction branches. As explained by Vishik \& Lyusternik \cite{VL}, the expansion of any of those branches may involve {\em fractional} powers of $k^2$. The problem of computing the coefficients of such expansions is interesting but,  for the sake of brevity, and because we have all that is needed for our immediate purpose, we shall not pursue it.

\subsection{The replica trick}
\label{quasiSolvableSubsection3d}
We find by direct calculation that, when $\ell$ is an even natural number, the operator $\Delta_K -k^2 {\mathscr A}$ admits the invariant finite-dimensional subspace
$$
\bigoplus_{l=0}^{\ell/2} \text{span} \left \{ \varepsilon_{lm} : \;0 \le m \le [l/2] \right \}\,.
$$

For $\ell=2$, the characteristic polynomial has the roots $\mu(k,2) = \mu_0 = - 8 k^2$ and $\mu_1 = -6 -8 k^2$.

For $\ell=4$, the characteristic polynomial factorises into the product of two quadratic polynomials; the leading root is
\begin{equation}
\notag
\mu(k,4) = \mu_0 = -20 k^{2}-10+2 \sqrt{36 k^{4}-12 k^{2}+25}\,.
\end{equation}
The reader will easily verify that the small-$k^2$ expansion is indeed obtained by setting $\ell=4$ in Formula (\ref{3dexpansionForMU0}); it converges for $k^2 \le 5/6$.
The next root of interest is
\begin{equation}
\notag
\mu_1 = -20 k^{2}-13+\sqrt{144 k^{4}-120 k^{2}+49}\,.
\end{equation}
Its small-$k^2$ expansion agrees with Formula (\ref{3dexpansionForMU1}); it converges for $k^2 \le 7/12$.
 
As in the two-dimensional case, we conjecture that the radius of convergence of the small-$k^2$ expansions decreases with $\ell$ and eventually dips below $1/3$.

\section{A renewing flow with intervals of pure strain}
\label{pureStrainSection}

In \S \ref{outlineSubsection}, we introduced the artificial parameter $k^2$ in order to take advantage of Proposition \ref{casimirProposition}, which relates the transfer operator to the Casimir of $\text{\rm SO}(d,{\mathbb R})$ when $k^2=1/d$.
Our purpose in the present section is to show that, for {\em every} $0 \le k^2 \le 1/d$, $\ell \mapsto \Lambda(k,\ell)$ may itself be interpreted as the generalised Lyapunov exponent of (the continuum limit of) a product of random matrices
associated with a renewing flow.

For simplicity, consider the case $d=2$. To define the renewing flow, it suffices to specify the velocity field in each of the subintervals $t_n < t < t_{n+1}$, where, for $d=2$, $t_n = 3 n \delta t$, $n \in {\mathbb N}$:
\begin{align*}
& \text{For $t_n < t < t_n + \delta t$} \,, \;\;  &\dot{x}_1(t) &= 0 ,\;\;  & \dot{x}_2 (t) &= u_{21} \left ( x_1 (t) + \eta_{21}^n \right )\,. \\
& \text{For $t_n + \delta t < t < t_n + 2 \delta t$}\,, \;\;   &\dot{x}_1(t) &= \alpha_{12}^n \,x_1 (t), & \dot{x}_2 (t) &= -\alpha_{12}^n \,x_2 (t)\,. \\
& \text{For $t_n + 2 \delta t < t < t_n + 3 \delta t$}\,,\;\; &\dot{x}_1(t) &= u_{12} \left ( x_2 (t) + \eta_{12}^n \right ), & \dot{x}_2 (t) &= 0\,.
\end{align*}
The novel element is the insertion of an interval consisting, in fluid-mechanical language, of a pure strain in the $x_1-x_2$ plane.
Here, the $\alpha_{12}^n$ denote independent draws from the distribution of a random variable with zero mean, independent of the $\eta_{ij}^n$. 

The Jacobians of the successive transformations $x(t_n) \mapsto x(t_n+\delta t)$, $x(t_n+\delta t) \mapsto x(t_n+2\delta t)$ and $x(t_n+2\delta t) \mapsto x(t_n+3\delta t)$ are given by the $2 \times 2$ matrices
$$
\begin{pmatrix}
1 & 0 \\
\delta t u_{21}'(x_1^n + \eta_{21}^n) & 1
\end{pmatrix},\;\;\begin{pmatrix}
e^{\delta t \alpha_{12}^n} & 0 \\
0 & e^{-\delta t \alpha_{12}^n}
\end{pmatrix}\;\;\text{and}\;\;\begin{pmatrix} 1 & \delta t u_{12}' ( x_2^{n+1} + \eta_{12}^n ) \\
0 & 1
\end{pmatrix}
$$
respectively, where $x^n := x (t_{n})$. Arguing as in \S \ref{flowToProductSubsection}, we are led to a product of the matrices
\begin{equation}
\notag
g_n := \begin{pmatrix} 1 & \delta t \,u_{12}' \left ( \phi_{12}^n  \right ) \\
0 & 1
\end{pmatrix}
\begin{pmatrix}
e^{\delta t \alpha_{12}^n} & 0 \\
0 & e^{-\delta t \alpha_{12}^n}
\end{pmatrix}
\begin{pmatrix}
1 & 0 \\
\delta t \,u_{21}' \left (\phi_{21}^n \right ) & 1
\end{pmatrix}
\end{equation}
where
$$
\phi_{12} := \eta_{12} + \delta t\, u_{21} \left ( \eta_{21} \right )\;\;\text{and}\;\;\phi_{21} = \eta_{21}\,.
$$

This construction generalises to any dimension, leading to products consisting of matrices of the form
\begin{multline}
g_n = \prod_{j \ne 1} e^{\delta t \xi_{1j}^n N_{1j}} \prod_{1<j} e^{\delta t \alpha_{1j}^n A_{1j}} \\
\cdots \prod_{j \ne d-1} e^{\delta t \xi_{d-1,j}^n N_{d-1,j}} \prod_{d-1<j} e^{\delta t \alpha_{d-1,j}^n A_{d-1,j}} \prod_{j \ne d} e^{\delta t \xi_{dj}^n N_{dj}}
\label{matrixWithStrain}
\end{multline}
where the $\xi_{ij}^n$ are given by the formulae (\ref{tij}-\ref{upperPhiAngles}) and the $\alpha_{ij}^n$ are mutually independent, zero-mean random variables independent of the $\xi_{ij}^n$.
In the continuum limit, the transfer operator associated with the product becomes
$$
1 + \frac{(\delta t)^2}{2} \left [ \sum_{i \ne j} {\mathbb E} \left ( \xi_{ij}^2 \right ) N_{ij}^2 + \sum_{i<j} {\mathbb E} \left ( \alpha_{ij}^2 \right ) A_{ij}^2 \right ]\,.
$$
If we now assume that
$$
{\mathbb E} \left ( \xi_{ij}^2 \right ) = \sigma^2\;\;\text{and}\;\; {\mathbb E} \left ( \alpha_{ij}^2 \right ) = \sigma^2  \left ( \frac{1}{d} - k^2 \right )
$$
for some positive $\sigma$ and some dimensionless number $0 \le k^2 \le 1/d$, independent of $i$ and $j$, then the corresponding generalised Lyapunov exponent is precisely $L(k,\ell)$.

With this model, the dimensionless number $1/d-k^2$  can be interpreted as a measure of the disorder strength in the pure-strain coefficients $\alpha_{ij}$, relative to that in the velocity gradients $u_{ij}'$. Plots of the first few cumulants
against $1/d-k^2$ are shown in Figure \ref{cumulantFigure}. We observe in particular that, regardless of the dimension, $\gamma_1$ and $\gamma_2$ increase with the relative strength in the pure-strain disorder, whereas $\gamma_3$ and $\gamma_4$ decrease.
Also noteworthy is the fact that, for $d=2$,
$$
\gamma_1-\gamma_2 = O \left ( k^4 \right ) \;\;\text{and}\;\;\gamma_3 - \gamma_4 = O \left ( k^8 \right ) \;\;\text{as $k \rightarrow 0$}\,.
$$
 
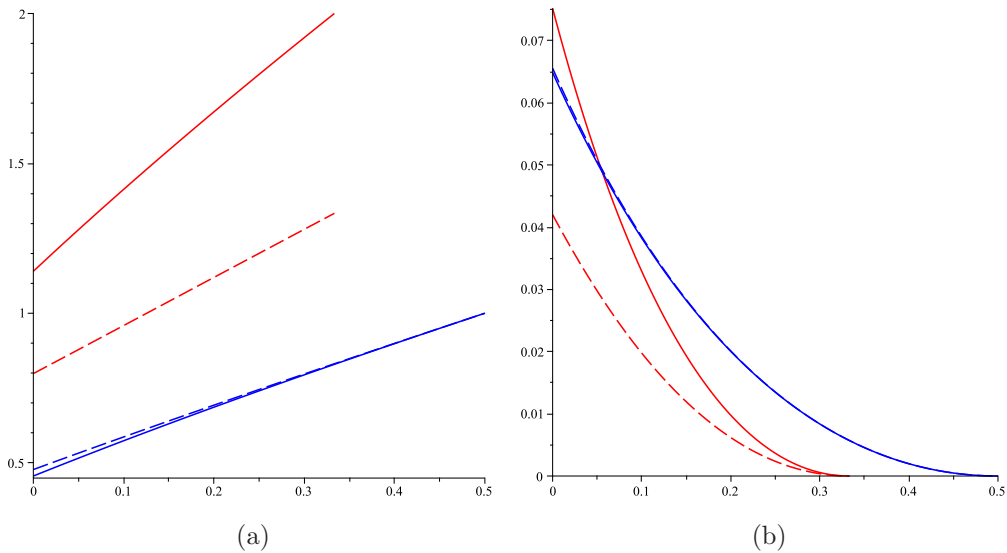
\begin{figure}[htbp]
  \centering
  \includegraphics[width=0.45\linewidth]{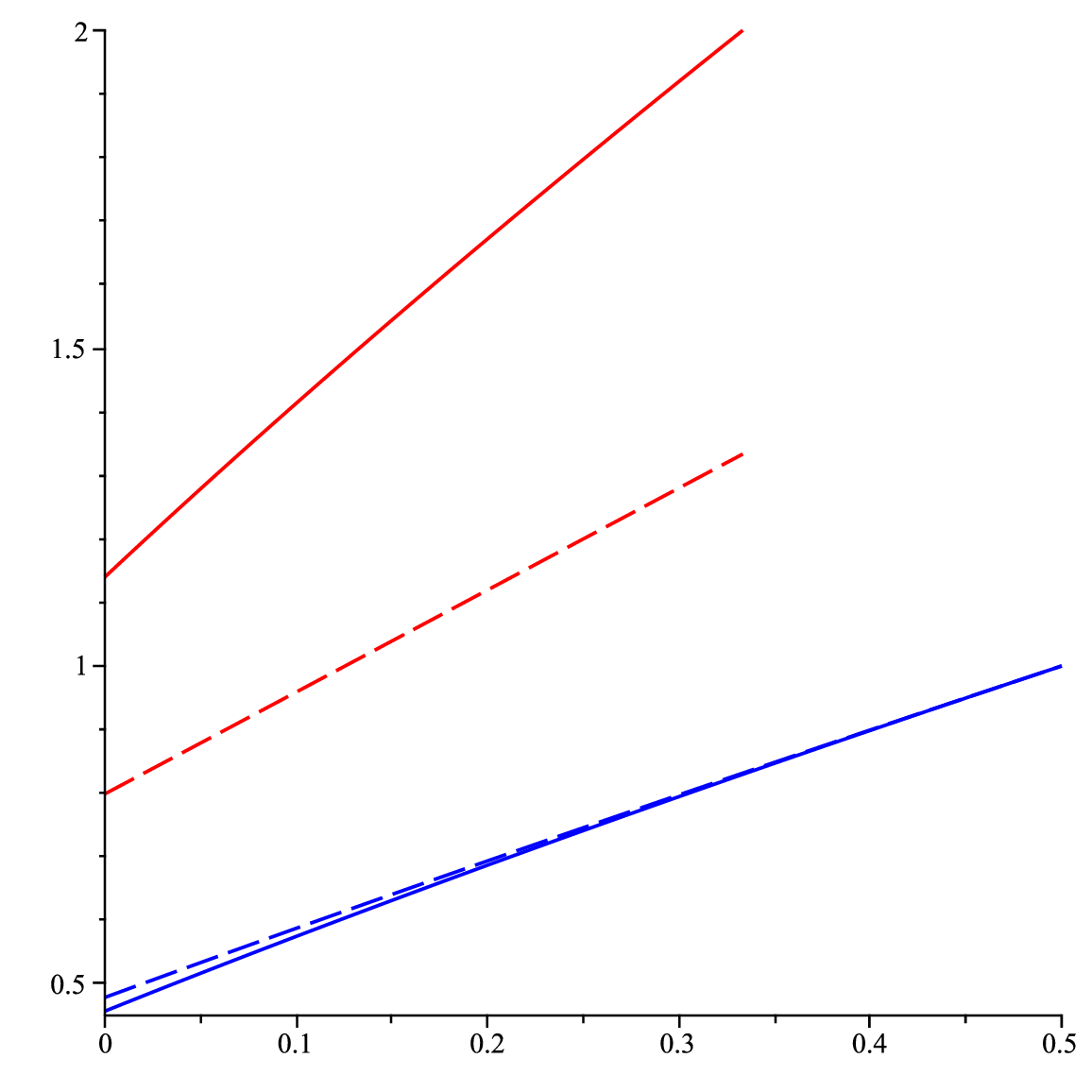}%
  \hspace{0.05\linewidth}%
  \includegraphics[width=0.45\linewidth]{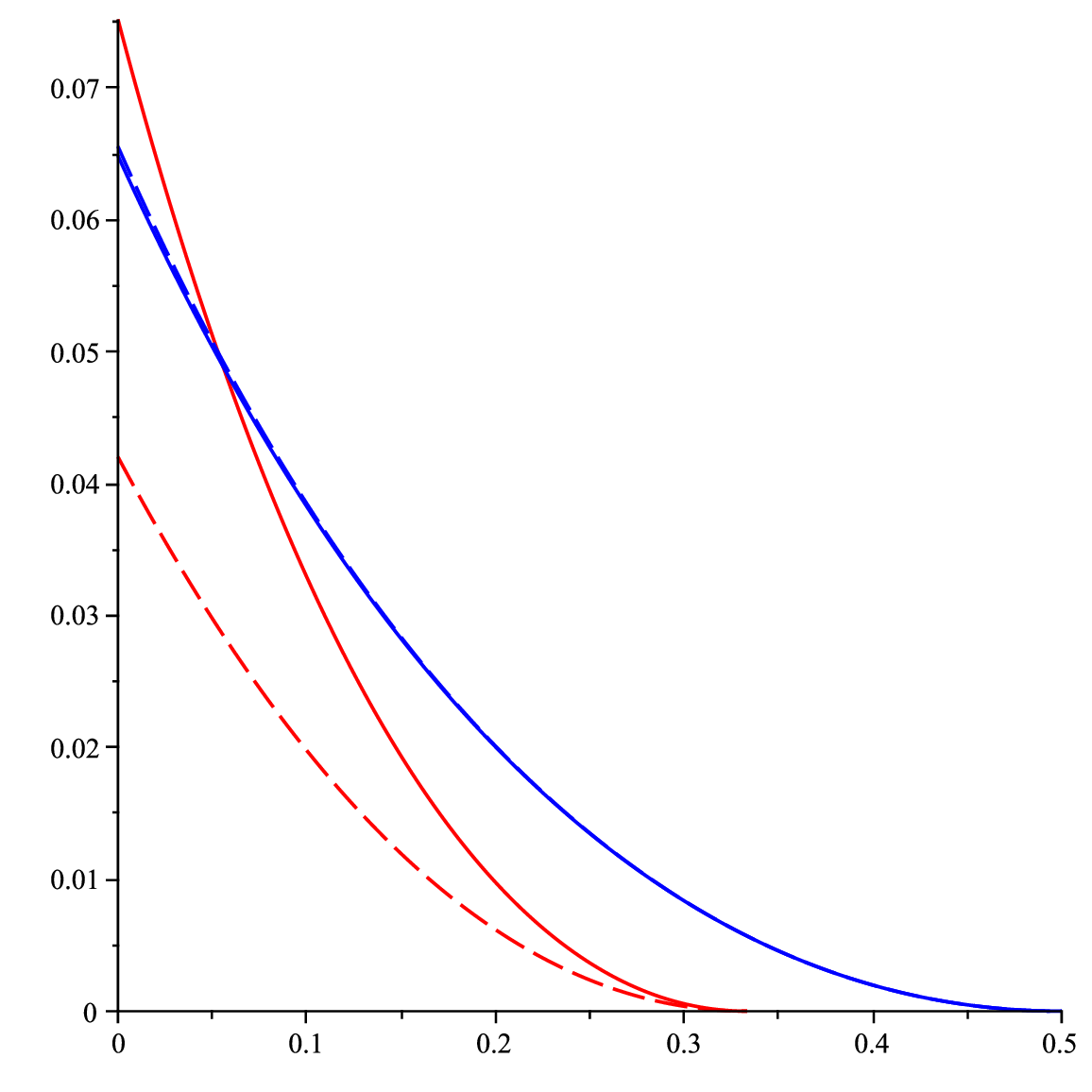}
    \begin{picture}(0,0) 
\put(-270,-7){(a)}
\put(-80,-7){(b)} 
\end{picture} 
  \caption{Plots of the cumulants $\gamma_j$ against $1/d-k^2$ for $1 \le j \le 4$. 
  The blue and red curves correspond to $d=2$ and $d=3$ respectively: 
  (a) $\gamma_1$ (solid curves) and $\gamma_2$ (dashed curves); 
  (b) $\gamma_3$ (solid curves) and $\gamma_4$ (dashed curves).}
  \label{cumulantFigure}
\end{figure}

The fact that, in the case $d=2$, the first two cumulants nearly coincide is reminiscent of the ``single parameter scaling'' property exhibited by other models; see \cite{Te}.


\section{Discussion and concluding remarks}
\label{concludingSection}

We have considered a product of matrices related to a class of $d$-dimensional incompressible renewing flows and studied its continuum limit with, as a practical objective, 
the calculation of the associated generalised Lyapunov exponent $L(\ell)$.

For the case $d=2$, the analysis involves a spectral problem that can be mapped, by allowing a certain parameter ratio to assume imaginary values, to that studied by Chetrite {\em et al.} \cite{CDG} for the Kraichnan ensemble of velocity fields. We computed, as they did, the first two cumulants 
in terms of elliptic functions. We also showed how these calculations relate, through the work of \cite{CLTT}, to the one-dimensional Anderson model at zero energy \cite{KW,DG,ST}, and to the one-dimensional Dirac equation with a random mass \cite{RT}.

For $d \ge 2$, we embedded the problem in a $k$-dependent family such that, when $k=0$, the problem is
solvable. Treating $k$ as a perturbation parameter, we then obtained expansions in ascending powers of $k^2$ for a function $\Lambda(k,\ell)$ that, suitably rescaled, coincides in the continuum limit with $L(\ell)$ when $k^2=1/d$.  The first few terms of this expansion were computed explicitly for $d \in \{2,3\}$. We observed  that the radius of convergence relates to the spectral gap separating the dominant eigenvalue of the transfer operator from the rest of its spectrum, and that it decreases with $\ell$ so that, when $k^2=1/d$, the expansions only converge for moderate values of $\ell$.  Nevertheless, they provided an effective means of computing the  cumulants, which depend only on the local behaviour of the generalised Lyapunov exponent at $\ell=0$.

We also showed that, for every $0 \le k^2 \le 1/d$, $\ell \mapsto \Lambda(k,\ell)$ is--- still in the continuum limit and suitably scaled--- the generalised Lyapunov exponent of a product of matrices associated with the renewing flow obtained by inserting time intervals of pure strain; the number $1/d-k^2$ then measures the strength of the disorder in the pure strain, relative to that in the pure-strain-free version. 
Whether there is, for the case $d >2$, a simple mapping between this model and the Kraichnan model studied by Chetrite {\em et al.} is an interesting open question.

\subsection{Comparison with the calculations of Haynes and Vanneste}
\label{haynesVannesteSection}
As mentioned in the introduction, Haynes \& Vanneste's objective is to determine conditions under which the decay rate of the covariance $\Gamma (0,t)$ of a passive scalar, evolving in a spatially periodic domain, may be deduced from the stretching characteristics of the random flow.  The main conclusion of the analysis presented in their section III, devoted to the continuum limit, is that, as long as the period of the spatial domain is a small enough integer multiple of that of the velocity field, then
$$
-\lim_{t \rightarrow \infty} \frac{\ln \Gamma ({\mathbf x},t)}{t} = \Lambda(0) - \frac{2 \pi^2 \Lambda''(0)}{\ln^2 \kappa} + o \left ( \frac{1}{\ln^2 \kappa} \right ) \;\;\text{as $\kappa \rightarrow 0+$}
$$
where $\Lambda(0)$ and $\Lambda''(0)$ are expressible in terms of what we denote in the present paper $\Lambda(k,\ell)$.
This is their Equation (3.8), reproduced here using their notation, which uses $\kappa$ instead of $D$ for the diffusivity constant. 
In the statement just above their equation (3.21), they also give the numerical values $\Lambda(0) \approx 0.226$ and $\Lambda''(0) \approx -0.45$ but,
when one compares their equations (3.8) and (3.21), it becomes clear that $\Lambda''(0) \approx -0.45$ was intended. 

The precise correspondence between their function $\Lambda$ and our own $\Lambda(k,\ell)$ is clouded by some unfortunate clashes of notation. To clarify it, let us take as starting point Haynes \& Vanneste's Equation (3.19), obtained after separating the spatial variable ${\mathbf x} = (x,y)$ from the temporal variable $t$:
$$
- \left ( y^2 \frac{\partial^2}{\partial x^2} + x^2 \frac{\partial^2}{\partial y^2} \right ) \Gamma = \lambda \Gamma\,.
$$
We recognise on the left-hand side the operator which we would denote  $-\left ( N_{21}^2+N_{12}^2 \right )$ in the basic realisation. Haynes and Vanneste look for a solution of the form (see their Equation (3.4)) 
$$
\Gamma ({\mathbf x}) = r^{\sigma-1} f_\sigma (\theta)\,.
$$
This is obviously equivalent to working with the Iwasawa realisation of the representation. Their equation (3.20) thus corresponds to our equation (\ref{perturbedSpectralProblemForDEqual2}). 
Let us agree to write
$$
A \sim B
$$
to signify that the mathematical object which Haynes and Vanneste denote by the symbol $A$ on the left-hand side is denoted in our paper by the symbol $B$ on the right-hand side. We then obtain the following ``dictionary'':
\begin{multline}
\notag
{\mathbf x} = (x,y) \sim x = (x_1,x_2)\,,\;\;\kappa \sim D\,,\;\;\sigma \sim \ell+1\,,\\
f_\sigma (\theta) \sim v(\theta)\,,\;\;
\Lambda (\sigma) \sim -\Lambda \left (k,\ell \right ) \Bigl |_{k=\frac{1}{\sqrt{2}}}\,.
\end{multline}

In particular, we see that, in the work of Haynes and Vanneste, it is the
value taken by  $\Lambda(1/\sqrt{d},\ell)$ at its minimum
\begin{equation}
\ell = \ell_{\min} := -\frac{d}{2}\,,
\label{Lmin}
\end{equation}
as well as the value of the second derivative there,
that govern the temporal decay of the covariance.

It turns out that $\ell_{\min}$ is small enough to ensure that the partial sums of the series for $\Lambda(k,\ell)$ in powers of $k^2$ converge nicely at $k^2=1/d$ when $\ell = \ell_{\min}$.
In particular, for $d=2$, we estimate--- on the basis of the first 64 terms computed with the help of the mathematical software {\tt Maple}--- that, for $\ell=\ell_{\min}=-1$, the radius of convergence is approximately $1$.
This expansion translates into
\begin{multline}
\Lambda \left ( k,\ell_{\min} \right ) = -\frac{1}{2}+\frac{1}{2} k^{2}+\frac{9}{128} k^{4} +
\frac{9}{256} k^{6}+\frac{47241}{2097152} k^{8}  \\
+\frac{67995}{4194304} k^{10} 
+\frac{13407669}{1073741824} k^{12}+\frac{21598857}{2147483648} k^{14}+\cdots
\label{expansionForLAtTheMinimum}
\end{multline}
and
\begin{multline}
\frac{\partial^2 \Lambda}{\partial \ell^2} \left ( k,\ell_{\min} \right ) = 1-k^{2}-\frac{5}{32} k^{4}-\frac{5}{64} k^{6}-\frac{13109}{262144} k^{8} \\ -\frac{18847}{524288} k^{10} 
-\frac{7424679}{268435456} k^{12}-\frac{11948355}{536870912} k^{14}+ \cdots
\label{expansionForLambdaDoublePrimeAtTheMinimum}
\end{multline}
For $k^2=1/2$, the partial sums converge well and, by comparing them, we deduce the following values, believed to be correct to the number of decimals displayed:
\begin{equation}
\notag
\Lambda \left (k, \ell_{\min} \right )  \Bigl |_{k= \frac{1}{\sqrt{2}}}= -0.2257817708 \ldots
\;\;\text{and}\;\;
\frac{\partial^2 \Lambda}{ \partial \ell^2} \left (k, \ell_{\min} \right ) \Bigl |_{k= \frac{1}{\sqrt{2}}} =  0.4461902388 \ldots
\end{equation}


\subsection{The spectrum of the transfer operator and its dependence on $\tau$, $k$ and $\ell$}
Equation (\ref{symmetricContinuumLimit}) and Proposition \ref{casimirProposition} imply that the eigenvalues of the transfer operator, in the continuum limit, are of the form
$$
\lambda = \lambda_l := 1 - \tau^2 l (l+d-2)\,,\;\;\text{for some}\;l \in \{0,1, 2,\ldots\},
$$
to leading order as $\ell$ and $k$ both tend to zero. In particular, the leading eigenvalue appearing in Equation (\ref{spectralCharacterisation}) is indeed strictly positive in this limiting case. What becomes of the spectrum as the parameters $k$, $\ell$ and $\tau$ 
vary is an open question; a systematic numerical study of the convergence properties of the expansions in powers of $k$ obtained in this paper could provide a useful starting point for investigating it.

\subsection{Other Lyapunov exponents}
We have motivated the introduction of the generalised Lyapunov exponent $L(\ell)$ as a means of quantifying how quickly nearby trajectories of a $d$-dimensional dynamical system separate. 
For $d>2$, the temporal growth of $m$-dimensional volumes, $2 \le m < d$, is also of obvious interest; see \cite{NV}. Tutubalin \cite{Tu} showed how to generalise the representation $T_\ell$ in order to define
a suitable transfer operator. However, it is not at present clear to us whether there exists, for these more complicated representations, a counterpart of Proposition \ref{casimirProposition}, that could relate the transfer operator to the Casimir of $\text{SO}(d,{\mathbb R})$
in a useful way.

\subsection{Beyond the continuum limit}
\label{beyondSubsection}

We end the paper by discussing, briefly and for a very concrete two-dimensional example, some of the issues that one encounters when attempting to look beyond $o(\tau^2)$. 

The first observation is that the relationship between the transfer operator and the angular Laplacian is lost, so that no particular advantage is gained by working with the Iwasawa realisation of the basic representation. Rather, because
the matrices in the product are themselves products of triangular matrices, it is more natural to work with another realisation, based on the {Gauss decomposition} of the group
$\text{SL} (d,{\mathbb R})$ \cite{Vi}. For $d=2$, the resulting realisation takes the form
\begin{equation}
\left [ T_{\ell} (g) v \right ] (z) = \left | a+z c\right |^{\ell} v ( z \cdot g )\;\;\text{where}\;\; z \cdot g := \frac{b+z d}{a+z c}\,.
\label{realisationOnTheLine}
\end{equation}
Here, $z \in {\mathbb R}$ is the counterpart of the angle used as independent variable in the Iwasawa realisation. The corresponding infinitesimal generators are displayed in the last column of Table \ref{2dInfinitesimalGeneratorTable}.

The particular choice
\begin{equation}
u_{12}' (\theta) = u_{21}' (\theta) = \sqrt{2} \,\sigma \cos \theta
\label{pierrehumbertModel}
\end{equation}
served as a toy model in the articles \cite{HV,Va}; it leads 
to  a  product $\Pi_n$ of independent draws from the distribution of
\begin{equation}
g = \begin{pmatrix} a & b \\ c & d \end{pmatrix} :=
\begin{pmatrix} 
1 & \alpha \\
0 & 1
\end{pmatrix} 
\begin{pmatrix}
1 & 0 \\
\beta & 1
\end{pmatrix} 
\label{productElement}
\end{equation}
where, according to Equation (\ref{2dRandomMatrix}), $\alpha = 2 \tau \cos \theta$ and $\beta= 2 \tau \cos \theta'$, with $\theta$ and $\theta'$ drawn independently from the uniform distribution in $(-\pi,\pi)$.

Let us take the calculation of the first cumulant $c_1$ in Equation (\ref{generalisedLyapunovExponent}) as a modest objective. The following result, whose proof we omit, is a simple application of Furstenberg's formula \cite{Fu} for the Lyapunov exponent  of a product of independent identically-distributed random matrices:
\begin{proposition}
Let $\nu (dz)$ be the probability measure of a real random variable $Z$, whose law is invariant under the action of the random matrix $g$ in Equation (\ref{productElement}), i.e.
\begin{equation}
Z \overset{\text{\rm law}}{=} \cfrac{1}{\beta+\cfrac{1}{\alpha + Z}} \,.
\label{continuedFraction}
\end{equation}
Then
$$
{c}_1 = \int_{\mathbb R} \nu (dz) \,{\mathbb E} \left ( \ln \left | \frac{\alpha+z}{z} \right | \right )\,.
$$
In particular, if $\alpha$ and $\beta$ are identically distributed, then
$$
{c}_1 = -2 \int_{\mathbb R} \nu(dz)\,\ln |z|\,. 
$$
\label{furstenbergProposition}
\end{proposition}

The problem of computing $c_1$ therefore reduces to that of finding the distribution of the random continued fraction
implied by Equation (\ref{continuedFraction}). This continued fraction has the same form as that considered by Dyson in his study of a random linear chain \cite{Dy,Fo}. There is, however, one significant novelty: in Dyson's case the $\alpha_j$ and $\beta_j$
represent random masses and Hooke constants, and so assume strictly positive values; by contrast, in the present case, the $\alpha_j$ and $\beta_j$ can assume negative values. This is problematic 
(and the situation gets worse as $\tau$ increases) because Equation (\ref{continuedFraction}) implies that the probability of dividing by a number in an interval centered on zero is proportional to the size of the interval.

In what follows, we shall {\em assume} that $\nu(dz)$ is absolutely continuous with respect to the Lebesgue measure, so that
\begin{equation}
\nu (dz) = f(z)\,dz
\label{probabilityDensity}
\end{equation}
for some unknown probability density function $f$ that is smooth, and see where this hypothesis leads us.

The smooth invariant density $f$ must belong to the representation space $V_{-2}$ and solve the so-called Dyson--Schmidt equation \cite{CTT}:
\begin{equation}
f = {\mathscr T}_0^\dag f = {\mathbb E} \left ( e^{-\beta N_{21}}  \right ) \circ {\mathbb E} \left ( e^{-\alpha N_{12}} \right ) f 
\label{vannesteDysonSchmidtEquation}
\end{equation}
where 
$$
N_{12} = \frac{d}{dz} \;\;\text{and}\;\; N_{21} = -\frac{d}{dz} z^2
$$ 
are the infinitesimal generators, listed in Table \ref{2dInfinitesimalGeneratorTable}, corresponding to the Gauss realisation with $\ell=-2$. It will be convenient, in the first instance, to normalise the solution by imposing
the condition
\begin{equation}
f(0) = 1\,.
\label{normalisationBeyondTheContinuumLimit}
\end{equation}
Equation (\ref{vannesteDysonSchmidtEquation}) may be expressed formally as
\begin{equation}
\notag
f = \sum_{j=0}^\infty \frac{{\mathbb E} \left ( \alpha^{j} \right )}{(j)!} N_{21}^{j} \circ \sum_{j=0}^\infty \frac{{\mathbb E} \left ( \beta^{j} \right )}{(j)!} N_{12}^{j} f\,.
\end{equation}
By using
$$
{\mathbb E} \left ( \frac{\alpha^m}{m!} \right )  = {\mathbb E} \left ( \frac{\beta^m}{m!} \right )
= \frac{\tau^m}{m!} \begin{cases} \binom{m}{m/2} & \text{if $m$ is even} \\
0 & \text{if $m$ is odd}
\end{cases}\,,
$$
we can eventually put it in the form
\begin{multline}
\notag
\frac{d}{dz} \left ( \frac{d}{dz} + \frac{d}{dz} z^2 \right ) f =: \left ( N_{12}^2 + N_{21}^2 \right ) f \\
= -\sum_{j=2}^\infty \tau^{2(j-1)} \sum_{k=0}^j \frac{1}{[k! (j-k)!]^2} N_{21}^{2(j-k)} N_{12}^{2k} f\,.
\end{multline}

Previously, we neglected the terms of $o ( \tau^2 )$ on the right-hand side. 
To go beyond this approximation, we look for a solution that 
 admits the asymptotic expansion
\begin{equation}
f(z) \sim \sum_{j=0}^\infty f_{j} (z) \,\tau^{2j} \;\;\text{as $\tau^2 \rightarrow 0+$}\,.
\label{asymptoticDensity}
\end{equation}
We then have the following explicit formulae for the first few coefficients:
$$
f_0 (z) := \frac{1}{\sqrt{1+z^4}}\,.
$$
$$
f_1 (z) = -\frac{5 z^{2}}{\left(z^{4}+1\right)^{7/2}}+\frac{11 z^{2}}{\left(z^{4}+1\right)^{5/2}}-\frac{9 z^{2}}{2 \left(z^{4}+1\right)^{3/2}}\,.
$$
\begin{multline}
\notag
f_2 (z) = -\frac{1155}{2 \left(z^{4}+1\right)^{13/2}}+\frac{4725}{2 \left(z^{4}+1\right)^{11/2}}-\frac{3605}{\left(z^{4}+1\right)^{9/2}} \\
+\frac{4985}{2 \left(z^{4}+1\right)^{7/2}}-\frac{5915}{8 \left(z^{4}+1\right)^{5/2}}+\frac{535}{8 \left(z^{4}+1\right)^{3/2}}\,.
\end{multline}
\begin{multline}
\notag
f_3(z) = \frac{425425 z^{2}}{2 \left(z^{4}+1\right)^{19/2}}-\frac{1076075 z^{2}}{\left(z^{4}+1\right)^{17/2}}+\frac{8826675 z^{2}}{4 \left(z^{4}+1\right)^{15/2}}\\
-\frac{9366005 z^{2}}{4 \left(z^{4}+1\right)^{13/2}}+\frac{10857703 z^{2}}{8 \left(z^{4}+1\right)^{11/2}}-\frac{1658003 z^{2}}{4 \left(z^{4}+1\right)^{9/2}} \\
+\frac{2784713 z^{2}}{48 \left(z^{4}+1\right)^{7/2}}-\frac{121535 z^{2}}{48 \left(z^{4}+1\right)^{5/2}}-\frac{5 z^{2}}{4 \left(z^{4}+1\right)^{3/2}}
\end{multline}
and so on.

With the help of the second formula in Proposition \ref{furstenbergProposition}, we can then compute the terms in the expansion 
$$
{c}_1  = \sum_{j=0}^\infty {c}_{1,j} \,\tau^{2j}
$$
of the first cumulant in powers of $\tau^2$.  The first few coefficients are 
\begin{equation}
\notag
{c}_{1,0} = 0\,,\;\;
{c}_{1,1} = \frac{1}{\pi} G^{-2}\,,\;\;
{c}_{1,2} = \frac{1}{24} + \frac{3}{8 \pi^2} G^{-4}\,,\;\;
{c}_{1,3} = \frac{5}{192 \pi} G^{-2} + \frac{9}{64 \pi^3} G^{-6}\,,
\end{equation}
where $G$ is the Gauss lemniscate constant; see Equation (\ref{gaussConstant}).

One can easily calculate many more terms. It turns out, however, that the series has a radius of convergence equal to zero, indicating that the limit $\tau \rightarrow 0+$ is singular, and casting some doubt on the existence of a smooth probability density.

\appendix

\section{Independence of the $\xi_{ij}$}
\label{propositionAppendix}
\begin{proposition}
Let the $u_{ij}$ be $2 \pi$-periodic, and let the $\eta_{ij}$, $1 \le i \ne j \le d$, be independent random variables distributed uniformly in $(-\pi,\pi)$. Then the joint characteristic function
$$
\chi (s) := {\mathbb E} \left [ \exp \left ( \text{\rm i} \sum_{i \ne j} s_{ij} \xi_{ij} \right ) \right ] 
$$
of the random variables $\xi_{ij} := u_{ij}' \left ( \phi_{ij} \right )$, $1 \le i \ne j \le d$, is given by the formula
$$
\chi (s) = \prod_{i \ne j} \chi_{ij} \left ( s_{ij} \right )
$$
where 
$$
\chi_{ij} \left ( s_{ij} \right ) := \frac{1}{2 \pi} \int_{-\pi}^\pi d \theta \exp \left [ \text{\rm i} s_{ij} u_{ij}' (\theta) \right ]\,.
$$
In particular, the $\xi_{ij}$ are independent.
\label{disorderProposition}
\end{proposition}

\begin{proof}
We begin by considering the case $d=2$. By definition,
$$
\phi_{21} = \eta_{21}\;\;\text{and}\;\;
\phi_{12} := \eta_{12} + \delta t\, u_{21} (\phi_{21})\,.
$$
Hence
\begin{multline}
\chi (s) := {\mathbb E} \left [ \exp \left ( \text{\rm i} s_{12} \xi_{12} + \text{\rm i} s_{21} \xi_{21} \right ) \right ] \\
= \frac{1}{2 \pi} \int_{-\pi}^\pi d \eta_{21} \,e^{\text{\rm i} s_{21} u_{21}' (\eta_{21})} \frac{1}{2 \pi} \int_{-\pi}^{\pi} d \eta_{12} \,e^{\text{\rm i} s_{12} u_{12}' ( \phi_{12})} \,.
\notag
\end{multline}

Now,
\begin{multline}
\notag
 \frac{1}{2 \pi} \int_{-\pi}^{\pi} d \eta_{12} \,e^{\text{\rm i} s_{12} u_{12}' ( \phi_{12})}  \overset{\underset{\downarrow}{\phi_{12} = \eta_{12} + \delta t\, u_{21} (\phi_{21})}}{=} \int_{-\pi +  \delta t\, u_{21} (\phi_{21})}^{\pi +  \delta t\, u_{21} (\phi_{21})}
 d \phi_{12} \, e^{\text{\rm i} s_{12} u_{12}' ( \phi_{12})} \\
 =  \int_{-\pi}^{\pi}
 d \phi_{12} \, e^{\text{\rm i} s_{12} u_{12}' ( \phi_{12})}
\end{multline}
since the integrand is $2\pi$-periodic. The desired result follows.

For the case $d=3$, we have 
$$
\xi_{ij} = u_{ij}' \left (  \phi_{ij} \right )\,,
$$
$$
\phi_{21} = \eta_{21}\,,\;\;\phi_{31} = \eta_{31}\,,\;\;\phi_{32} = \eta_{32}
$$
and
\begin{align*}
\phi_{13} &=  \eta_{13} + \delta t \,u_{31} \left ( \phi_{31} \right ) + \delta t  \,u_{32} \left ( \phi_{32} \right ) \\
\phi_{23} &=  \eta_{23} + \delta t \,u_{31} \left ( \phi_{31} \right ) + \delta t \,u_{32} \left ( \phi_{32}\right ) \\
\phi_{12} &= \eta_{12} + \delta t \,u_{21} \left ( \phi_{21} \right ) + \delta t \,u_{23} \left ( \phi_{23} \right )\,.
\end{align*}
Hence
\begin{multline}
\notag
\chi (s) := \prod_{i>j} \frac{1}{2 \pi} \int_{-\pi}^\pi d \eta_{ij} \,e^{\text{\rm i} s_{ij} u_{ij}' ( \eta_{ij})} \\
\times \frac{1}{2 \pi} \int_{-\pi}^\pi d \eta_{13} \,e^{\text{\rm i} s_{13} u_{13}' ( \phi_{13})}   \frac{1}{2 \pi} \int_{-\pi}^\pi d \eta_{23} \,e^{\text{\rm i} s_{23} u_{23}' ( \phi_{23})} \frac{1}{2 \pi} \int_{-\pi}^\pi d \eta_{12} \,e^{\text{\rm i} s_{12} u_{12}' ( \phi_{12})}   \,.
\end{multline}
For the innermost integral, we have,
\begin{multline}
\notag
 \frac{1}{2 \pi} \int_{-\pi}^{\pi} d \eta_{12} \,e^{\text{\rm i} s_{12} u_{12}' ( \phi_{12})}  \\
 \overset{\underset{\downarrow}{\phi_{12} = \eta_{12} + \delta t\, u_{21} (\phi_{21})+ \delta t\, u_{23} (\phi_{23})}}{=} \int_{-\pi +  \delta t\, u_{21} (\phi_{21})+ \delta t\, u_{23} (\phi_{23})}^{\pi +  \delta t\, u_{21} (\phi_{21})+ \delta t\, u_{23} (\phi_{23})}
 d \phi_{12} \, e^{\text{\rm i} s_{12} u_{12}' ( \phi_{12})} \\
 =  \int_{-\pi}^{\pi}
 d \phi_{12} \, e^{\text{\rm i} s_{12} u_{12}' ( \phi_{12})}
\end{multline}
since the integrand is $2\pi$-periodic. For the same reason, the substitution $\phi_{i3} =   \eta_{i3} + \delta t \,u_{31} \left ( \phi_{31} \right ) + \delta t  \,u_{32} \left ( \phi_{32} \right )$ produces
$$
\frac{1}{2 \pi} \int_{-\pi}^\pi d \eta_{i3} \,e^{\text{\rm i} s_{i3} u_{i3}' ( \phi_{i3})} = \frac{1}{2 \pi} \int_{-\pi}^\pi d \phi_{i3} \,e^{\text{\rm i} s_{i3} u_{i3}' ( \phi_{i3})}\,,\;\;i \in \{1,\,2\}\,,
$$
and the desired result follows. The same line of proof may be used for other values of $d$, but the details are tedious.
\end{proof}

\section{Formulae for the $A_{ij} \,e_{lm}$}
\label{3dAppendix}
We introduce the sequences
\begin{equation}
\notag
a_{l} = \frac{\ell+2l+1}{4 (4l-1)(4l+1)}\,,\;\;
b_{l} = -\frac{2 \ell+3}{4(4l-1)(4l+3)}\,,\;\;
c_{l} = \frac{\ell-2l}{4(4l+1)(4l+3)} \,,
\end{equation}
\begin{multline}
\notag
{\mathfrak a}_{lm} = -6 \,a_l \left ( 2 l + 2m -1 \right ) \left ( 2 l+2m \right ), \;\;
{\mathfrak b}_{lm} = 4 \,b_l \left ( 2 l^2+l-6m^2\right ),\\
{\mathfrak c}_{lm} = -6\,c_l \left ( 2 l - 2m+1 \right ) \left ( 2l -2 m+2 \right ),
\end{multline}
\begin{multline}
\notag
\alpha_{lm} = a_l \frac{(2l+2m)!}{(2l+2m-4)!}\,,\;\;
\beta_{lm} = b_l \frac{(2l+2m)!}{(2l-2m)!} \frac{(2l-2m+2)!}{(2l+2m-2)!}\,, \\
\gamma_{lm} = c_l \frac{(2l-2m+4)!}{(2l-2m)!}\,.
\end{multline}

With the help of the computer algebra software {\tt Maple}, we then obtain  the following formulae, conjectured on the basis of explicit calculations using moderate values of $l$, $m$, and then confirmed by testing over a wide range of values.

\subsection{Formula for $A_{12} \,e_{lm}$}
\label{a12Subsection}
$$
A_{12} \,e_{l0} = \begin{cases}
4 \,c_0 \,e_{11} & \text{if $l=0$} \\
4 \left ( b_1 \,e_{11} + c_1\,e_{21} \right ) & \text{if $l=1$} \\
4 \left ( a_l \, e_{l-1,1} + b_l \,e_{l1} +c_l \,e_{l+1,1} \right ) & \text{if $l>1$}
\end{cases}
$$
\begin{multline}
\notag
A_{12} \,e_{lm} = 
2 \left ( \alpha_{lm} \,e_{l-1,m-1}  + \beta_{lm} \,e_{l,m-1}   + \gamma_{lm}\,e_{l+1,m-1} \right . \\
\left . +a_l  \,e_{l-1,m+1} + b_l\,e_{l,m+1} +c_l \,e_{l+1,m+1} \right )\,,\;\;\text{if $0 < m < l-1$}\,,
\end{multline}
\begin{multline}
\notag
A_{12} \,e_{lm} = 
2 \left ( \alpha_{lm} \,e_{l-1,m-1}  + \beta_{lm} \,e_{l,m-1}   + \gamma_{lm}\,e_{l+1,m-1} \right . \\
\left . +b_l \,e_{l,m+1} + c_l \,e_{l+1,m+1} \right ) \,,\;\;\text{if $0 < m = l-1$}\,,
\end{multline}
and, if $0< m=l$,
\begin{equation}
\notag
A_{12} \,e_{lm} = 
2 \left ( \alpha_{lm} \,e_{l-1,m-1}  + \beta_{lm} \,e_{l,m-1}   + \gamma_{lm}\,e_{l+1,m-1} 
+c_l \,e_{l+1,m+1} \right ) \,.
\end{equation}

\subsection{Formula for $A_{13} \,e_{lm}$}
\label{a13Subsection}
$$
A_{13} \,e_{l0} = \begin{cases}
{\mathfrak c}_{00}\, e_{10} + 2 \,c_0 \,e_{11} & \text{if $l=0$} \\
{\mathfrak a}_{10} \,e_{00}+{\mathfrak b}_{10} \,e_{10}+{\mathfrak c}_{10} \,e_{20}+ 2 \left ( b_1 \,e_{11} + c_1\,e_{21} \right ) & \text{if $l=1$} \\
{\mathfrak a}_{l0} \,e_{l-1,0}+{\mathfrak b}_{l0} \,e_{l0}+{\mathfrak c}_{l0} \,e_{l+1,0}+ 2 \left ( a_l \, e_{l-1,1} + b_l \,e_{l1} +c_l \,e_{l+1,1} \right ) & \text{if $l>1$}
\end{cases}
$$
\begin{multline}
\notag
A_{13} \,e_{lm} = 
\alpha_{lm} \,e_{l-1,m-1}  + \beta_{lm} \,e_{l,m-1}   + \gamma_{lm}\,e_{l+1,m-1} \\
+ {\mathfrak a}_{lm} \,e_{l-1,m} + {\mathfrak b}_{lm} \,e_{lm} + {\mathfrak c}_{lm} \,e_{l+1,m}  \\
+a_l  \,e_{l-1,m+1} + b_l\,e_{l,m+1} +c_l \,e_{l+1,m+1}\,,\;\;\text{if $0 < m < l-1$}\,,
\end{multline}
\begin{multline}
\notag
A_{13} \,e_{lm} = 
\alpha_{lm} \,e_{l-1,m-1}  + \beta_{lm} \,e_{l,m-1}   + \gamma_{lm}\,e_{l+1,m-1}  \\
+ {\mathfrak a}_{lm} \,e_{l-1,m} + {\mathfrak b}_{lm} \,e_{lm} + {\mathfrak c}_{lm} \,e_{l+1,m}  
+b_l \,e_{l,m+1} + c_l \,e_{l+1,m+1} \,,\;\;\text{if $0 < m = l-1$}\,,
\end{multline}
and
\begin{multline}
\notag
A_{13} \,e_{lm} = 
\alpha_{lm} \,e_{l-1,m-1}  + \beta_{lm} \,e_{l,m-1}   + \gamma_{lm}\,e_{l+1,m-1} \\
+ {\mathfrak b}_{lm} \,e_{lm} + {\mathfrak c}_{lm} \,e_{l+1,m} + c_l \,e_{l+1,m+1} \,,\;\;\text{if $0 < m = l$}\,.
\end{multline}

\subsection{Formula for $A_{23} \,e_{lm}$}
\label{a23Subsection}
$$
A_{23} \,e_{l0} = \begin{cases}
{\mathfrak c}_{00}\, e_{10} - 2 \,c_0 \,e_{11} & \text{if $l=0$} \\
{\mathfrak a}_{10} \,e_{00}+{\mathfrak b}_{10} \,e_{10}+{\mathfrak c}_{10} \,e_{20}- 2 \left ( b_1 \,e_{11} + c_1\,e_{21} \right ) & \text{if $l=1$} \\
{\mathfrak a}_{l0} \,e_{l-1,0}+{\mathfrak b}_{l0} \,e_{l0}+{\mathfrak c}_{l0} \,e_{l+1,0}- 2 \left ( a_l \, e_{l-1,1} + b_l \,e_{l1} + c_l \,e_{l+1,1} \right ) & \text{if $l>1$}
\end{cases}
$$
\begin{multline}
\notag
A_{23} \,e_{lm} = 
-\alpha_{lm} \,e_{l-1,m-1}  - \beta_{lm} \,e_{l,m-1}   - \gamma_{lm}\,e_{l+1,m-1} \\
+ {\mathfrak a}_{lm} \,e_{l-1,m} + {\mathfrak b}_{lm} \,e_{lm} + {\mathfrak c}_{lm} \,e_{l+1,m}  \\
- a_l  \,e_{l-1,m+1} - b_l\,e_{l,m+1} - c_l \,e_{l+1,m+1}\,,\;\;\text{if $0 < m < l-1$}\,,
\end{multline}
\begin{multline}
\notag
A_{23} \,e_{lm} = 
-\alpha_{lm} \,e_{l-1,m-1}  - \beta_{lm} \,e_{l,m-1} - \gamma_{lm}\,e_{l+1,m-1}  \\
+ {\mathfrak a}_{lm} \,e_{l-1,m} + {\mathfrak b}_{lm} \,e_{lm} + {\mathfrak c}_{lm} \,e_{l+1,m}  
- b_l \,e_{l,m+1} - c_l \,e_{l+1,m+1} \,,\;\;\text{if $0 < m = l-1$}\,,
\end{multline}
and
\begin{multline}
\notag
A_{23} \,e_{lm} = 
-\alpha_{lm} \,e_{l-1,m-1}  - \beta_{lm} \,e_{l,m-1} - \gamma_{lm}\,e_{l+1,m-1} \\
+ {\mathfrak b}_{lm} \,e_{lm} + {\mathfrak c}_{lm} \,e_{l+1,m} - c_l \,e_{l+1,m+1} \,,\;\;\text{if $0 < m = l$}\,.
\end{multline}

\bibliographystyle{amsplain}

\begin{thebibliography}{10}



\bibitem{Ap} Apostol, T. M.: Modular Functions and Dirichlet Series in Number Theory. Springer--Verlag, New York (1976)


\bibitem{CDG} Chetrite, R., Delannoy, J.-Y.,  Gaw\c{e}dzki, K.: Kraichnan flow in a square: an example of integrable chaos. J. Stat. Phys. 126, 1165–1200 (2007)

\bibitem{CG} Childress, S., Gilbert, A. D.: Stretch, Twist, Fold: The Fast Dynamo. Springer--Verlag, Heidelberg (1995)
 
\bibitem{CLTT}  Comtet, A., Luck, J.-M., Texier, C.,Tourigny, Y.: The Lyapunov exponent of products of random 2×2 matrices close to the identity. J. Stat. Phys. 150, 13–65 (2013)

\bibitem{CTT} Comtet, A., Texier, C.,Tourigny, Y.: 
Representation theory and products of random matrices in $\text{SL}(2,{\mathbb R})$.
https://arxiv.org/pdf/1911.00117

\bibitem{CTT1} Comtet, A., Texier, C., Tourigny, Y.: Generalised Lyapunov exponent for the one-dimensional Schr\"{o}dinger equation with Cauchy disorder: exact results. Phys. Rev. E 105 064210  (2022)

\bibitem{CFV} Crisanti, A., Falcioni, M., Vulpiani A., Paladin, G.: Lagrangian chaos: transport, mixing and diffusion in fluids. {Riv. del Nuovo Cimento} 14, 1-80 (1991)

\bibitem{CPV} Crisanti, A., Paladin, G., Vulpiani, A.: Products of Random Matrices
in Statistical Physics. Springer--Verlag, Heidelberg (1993)

\bibitem{De} Derrida, B.: Non-equilibrium steady states: fluctuations and
large deviations of the density and of the current. J. Stat. Mech. P07023 (2007) 

\bibitem{DG} Derrida, B., Gardner, E.: Lyapounov exponent of the one dimensional Anderson model:
weak disorder expansions. J. Phys. (Paris) 45,1283-1295 (1984).

\bibitem{Dy} Dyson, F. J.: The dynamics of a disordered chain. Phys. Rev. {92}, 1331-1338 (1953)

\bibitem{Fo} Forrester,  P. J.: Dyson’s disordered linear chain from a random matrix theory viewpoint. {J. Math. Phys.}  {62}, 103302  (2021)


\bibitem{Fu} Furstenberg, H.: Noncommuting random products. {Trans. Amer. Math. Soc.} {108}, 377-428  (1963)

\bibitem{FT} Furstenberg, H., Tzkoni, I.: Spherical functions and integral geometry.
Israel J. Math. {20}, 327-338  (1971)

\bibitem{Ha} Halperin, B. I.: Green's functions for a particle in a one-dimensional random potential. Phys. Rev. 139, A104 (1965)

\bibitem{HV} Haynes, P. H., Vanneste, J.: What controls the decay of passive scalars in smooth flows? Phys. Fluids {17}, 097103 (2005)

\bibitem{KW} Kappus, M., Wegner, F.: Anomaly in the band centre of the one-dimensional Anderson model. Z. Phys. B 45, 15-21 (1981)

\bibitem{La} Lawden, D. F.: {Elliptic Functions and Applications}. Springer--Verlag, New York (1989)

\bibitem{Ll} Lloyd, P.: Exactly solvable model of electronic states
in a three-dimensional disordered Hamiltonian: non-
existence of localized states. J. Phys. C: Solid St. Phys. 2, 1717–1725 (1969).


\bibitem{NV} Ngan, K., Vanneste, J.: Scalar decay in a three-dimensional flow. {Phys. Rev. E} {83}, 056306 (2011)

\bibitem{Ob} Oberhettinger, F. : {Fourier expansions}. Academic Press, New York (1973)

\bibitem{DLMF} Olver, F. W. J., Olde Daalhuis, A. B., Lozier, D. W., Schneider, B. I., Boisvert, R. F., Clark, C. W., Miller, B. R., Saunders, B. V., Cohl, H. S., McClain, M. A. eds.: {NIST Digital Library of Mathematical Functions}. http://dlmf.nist.gov/, Release 1.1.5 of 2022-03-15

\bibitem{RT} Ramola, K., Texier, C.: Fluctuations of random matrix products and 1D Dirac
equation with random mass. J. Stat. Phys. 157, 497–514 (2014)

\bibitem{ST} Schomerus, H., Titov, M.: Band-center anomaly of the conductance distribution in one-dimensional Anderson localization. Phys. Rev. B, 100201(R) (2003)

\bibitem{Te} Texier, C.: Generalized Lyapunov exponent of random matrices and universality classes for SPS in 1D Anderson localisation. Europhys. Lett. 131, 17002 (2020)

\bibitem{TH} Touchette, H., Harris, R. J.: Large deviation approach to nonequilibrium systems. In Klages, R., Just, W.,  Jarzynski, C., editors, Nonequilibrium Statistical Physics of Small Systems: Fluctuation Relations and Beyond, pages 311--335,  Wiley, New York (2013)

\bibitem{Tu} Tutubalin, V. N.: On limit theorems for the product of random matrices, {Theor. Prob. Appl.} {10}, 15-27 (1965)

\bibitem{Va} Vanneste, J.: Estimating generalized Lyapunov exponents for products of random matrices,
{Phys. Rev. E} {81}, 036701 (2010)

\bibitem{Vi} Vilenkin, N. Ja.: {Special Functions and the Theory of Group Representations}. American Mathematical Society, Providence (1968)


\bibitem{VL} Vishik, M. I., Lyusternik, L. A.: The solution of some perturbation problems for matrices and selfadjoint or non-selfadjoint differential equations I,
{Russ. Math. Surv.}  {15}, 1-73 (1960)

\bibitem{ZMRS} Zeldovich, Ya. B., Molchanov, S. A., Ruzmaikin, A. A., Sokoloff, D. D.: Intermittency, diffusion and generation in a nonstationary random medium, {Sov. Sci. Rev. C Math. Phys.} {7}, 1-110 (1988)

\end{thebibliography}

\end{document}